\documentclass[a4paper,11pt]{article}
\usepackage[margin=1in]{geometry}
\newcommand{\mysize}{\small}

\usepackage[cmex10]{amsmath}
\usepackage{amsthm}
\usepackage{latexsym}
\usepackage{algpseudocode}
\usepackage{algorithm}
\usepackage{array}
\usepackage{url}
\usepackage{tikz}
\usetikzlibrary{backgrounds,positioning,calc}
\begin{document}
\title{\textbf{The Power of Local Information in PageRank}\footnote{This is the full version of the WWW 2013 poster~\cite{Bressan&2013b}. Last update: July 8th, 2013.}}

\author{
{\mysize Marco Bressan}\\
{\mysize Dip.\ Informatica}\\
{\mysize Sapienza Universit\`a di Roma}\\
{\mysize Roma, Italy}\\
{\mysize bressan@di.uniroma1.it}
\and
{\mysize Enoch Peserico}\\
{\mysize Dip.\ Ing.\ Informazione}\\
{\mysize Universit\`a di Padova}\\
{\mysize Padova, Italy}\\
{\mysize enoch@dei.unipd.it}
\and
{\mysize Luca Pretto}\\
{\mysize Dip.\ Ing.\ Informazione}\\
{\mysize Universit\`a di Padova}\\
{\mysize Padova, Italy}\\
{\mysize pretto@dei.unipd.it}
}

\date{}

\newtheorem{theorem}{Theorem}
\newtheorem{lemma}{Lemma}
\newtheorem{observation}{Observation}
\newtheorem{definition}{Definition}
\newtheorem{property}{Property}
\newcommand{\normone}[1]{\lVert #1 \rVert}
\newcommand{\normtwo}[1]{\lVert #1 \rVert_2}
\newcommand{\minset}[1]{\operatorname{minset} #1 }
\newcommand{\rankset}[1]{\operatorname{rankset} #1 }
\newcommand{\frontier}[1]{\operatorname{frontier} #1 }
\newcommand{\kernel}[1]{\operatorname{kernel} #1 }
\newcommand{\links}{\emph{links()} }
\newcommand{\linksNoSp}{\emph{links()}}
\newcommand{\jump}{\emph{jump()} }
\newcommand{\jumpNoSp}{\emph{jump()}}
\newcommand{\crawl}{\emph{crawl()} }
\newcommand{\crawlNoSp}{\emph{crawl()}}
\newcommand*{\LJ}{LiveJournal}
\newcommand*{\IT}{\texttt{.it}}
\newcommand*{\figsize}{0.75}
\newcommand*{\subfigsize}{0.75}
\newcommand*{\GREXT}{pdf}
\algloopdefx{NoEndIf}[1]{\textbf{if} #1 \textbf{then}}

\maketitle

\begin{abstract}
How large a fraction of a graph must one explore to rank a small set of nodes according to their PageRank scores?
We show that the answer is quite nuanced, and depends crucially on the interplay between the correctness guarantees one requires and the way one can access the graph.
On the one hand, assuming the graph can be accessed only via ``natural'' exploration queries that reveal small pieces of its topology, we prove that deterministic and Las Vegas algorithms must in the worst case perform $n - o(n)$ queries and explore essentially the entire graph, independently of the specific types of query employed.
On the other hand we show that, depending on the types of query available, Monte Carlo algorithms can perform asymptotically better: if allowed to both explore the local topology around single nodes and access nodes at random in the graph they need $\Omega(n^{2/3})$ queries in the worst case, otherwise they still need $\Omega(n)$ queries similarly to Las Vegas algorithms.
All our bounds generalize and tighten those already known, cover the different types of graph exploration queries appearing in the literature, and immediately apply also to the problem of approximating the PageRank score of single nodes.
\\[.01\textheight]
\end{abstract}
\thispagestyle{empty}
\clearpage

\section{Introduction}
\label{sec:intro}
Suppose one has to compute the relative ranking induced by PageRank on a small set of nodes of a large, hardly
accessible graph.
Which fraction of the graph must be explored to assess this relative ranking, and to what extent do different
exploration primitives influence the exploration cost?
We investigate the interplay between algorithms for this ``local PageRank problem'' and graph exploration models,
focusing on the correctness guarantees of the former and on the locality properties of the latter.
For algorithms that guarantee always a correct output (i.e.\ Las Vegas and deterministic algorithms), we show that no
``natural'' exploration model exempts from exploring almost the entire input graph in the worst case. In particular, we
prove lower bounds on the number of nodes to visit that are tighter than all the existing ones, going as far as 
$n - o(n)$ for graphs on $n$ nodes; and that hold simultaneously under any exploration model, for any single algorithm
execution, and for any value of the PageRank scores and of their relative distance. 
The cornerstone of such a level of generality is the novel notion of \emph{ranking subgraph}, which characterizes the
cost of any algorithm returning always a correct result.
While we exploit this fact to compute tight worst-case lower bounds, we also show that, in general, even approximating
them within a factor $o(\log(n))$ is an NP-complete problem.
On the other hand, for algorithms that can return an incorrect output (i.e.\ Monte Carlo algorithms), we show that the
interplay with the exploration model is crucial: while under any ``strictly local'' exploration model they still undergo
worst-case lower bounds similar to those of Las Vegas algorithms, they turn to be very efficient with just a
``minimum amount'' of non-locality.
More precisely we prove that, if one can only visit neighbours of already visited nodes, then a worst-case $\Omega(n)$
lower bound applies, for any possible value of the PageRank scores and of their relative distance. If instead one is
allowed to jump to random nodes in the graph, then the exploration cost can be radically reduced. In this case, using
just two spare exploration primitives, we build an algorithm that ranks an arbitrary subset of nodes with
arbitrarily small error probability by visiting a portion of the graph inversely proportional to the smallest of
their scores -- thus ranging between $O(1)$ and $O(n)$. 
By proving (the first) lower bounds for Monte Carlo algorithms holding under any exploration model, we show that our
algorithm is optimal also among those employing more informative primitives, except for very small scores: in this case,
an even more efficient (and fully sublinear?)\ local algorithm for computing the PageRank ranking might exist.

While strengthening the existing bounds~\cite{Bar-Yossef&2008b,Bressan&2011}, our results delineate a clear separation
in cost between algorithms that guarantee always a correct solution, and algorithms that accept a small probability  of
error -- but only when these latter can access the graph in a ``global'' fashion.
Incidentally, this justifies why all the existing algorithms for local PageRank computations with worst-case cost
guarantees have a positive probability of error and rely on random access to the graph.
Although devised for the problem of computing locally the PageRank ranking, our lower and upper bounds apply also to the
problem of approximating locally the PageRank scores; the two problems appear then essentially equivalent,
answering an open question of~\cite{Bar-Yossef&2008b}.

\paragraph{Motivations and Related Work.}
\label{par:INTRO_related}
PageRank~\cite{BP98} is a prominent node centrality measure, proving impressively successful in tasks as diverse as
crawl seeding~\cite{CGMP98}, spam detection~\cite{Gyongyi&2004}, graph partitioning~\cite{Andersen&2006}, trendsetter
identification~\cite{Baeza&2012} and bibliometrics~\cite{Yan07}.
Indeed, PageRank has been rated among the top $10$ algorithms in data mining~\cite{Wu&2007}, and its large and growing
success has attracted a continuous stream of research.
While its algebraic properties have been explored intensely in the past (see~\cite{Bahmani&2010} and~\cite{LM04} for an
overview), in the last years the focus has moved to the algorithmic side, and especially on the issues of its parallel,
distributed, incremental, and local computation; we address this latter here. 

In this paper, by ``locality'' we mean the possibility of computing the output of a problem by examining a
limited part of the input. When the input is a graph, this typically means that an algorithm should
explore a small portion of it, starting from a given set of nodes and obtaining information on the
global structure one piece at a time, each time incurring a cost.
Locality, in this sense, has been studied since the eighties, especially with regard to classical problems on graphs
(see e.g.~\cite{Deng&1990,Papadimitriou&1989,Kalyanasundaram&1994}).
More recently, the difficulty to store, ``snapshot'', or even access in their entirety massive graphs (such as the web
or social networks) has urged to investigate locality as one of the crucial issues for the wide range of
computational problems that naturally arise in large graphs; and this has led to the introduction of models and
algorithms for solving (global) graph problems in a local fashion.
Notably,~\cite{Brautbar&2010} introduced the ``jump and crawl'' graph exploration model, along with local algorithms for
finding vertices with extreme topological properties (e.g.\ with high clustering coefficient),
while~\cite{Borgs&2012b,Borgs&2012c} investigated local algorithms for three graph problems -- finding the node
of maximum degree, a path between any two nodes, and a minimum dominating set.

The problem of estimating locally the PageRank score of a given node was introduced by Chen et al.\ in \cite{CGS04}, together with heuristic exploration strategies which give good score approximations by visiting only a handful of nodes.
From the theoretical point of view, however, under their model any algorithm requires to visit a worst-case $\Omega(n)$ nodes, as shown by our lower bounds.
In~\cite{Fogaras&2005}, Fogaras et al.\ indirectly lower bound the number of nodes one must explore to approximate
Personalized PageRank scores~\cite{H02}. Since (as they note) by a simple reduction this allows to compute the standard PageRank scores, our lower bounds also apply; and, indeed, special cases of their bounds match special cases of ours.
Avrachenkov et al.~\cite{Avrachenkov&2007} describe several variants of Monte Carlo score approximation algorithms based on random walks (on which leverages the algorithm of Section~\ref{sec:montecarlo}); however, as they require visiting $\Omega(n)$ nodes, they are not apt to local computations.
The first lower bounds on the cost of locally approximating the score of a node within a factor $1+\epsilon$ are given by Bar-Yossef et al.\ in \cite{Bar-Yossef&2008b}, in a scenario where algorithms are not limited to crawl
back from the target node (differently from~\cite{CGS04}), but still under a specific exploration model.
Their $\Omega(n)$ bound for deterministic algorithms and $\Omega(\sqrt{n})$ bounds for Las Vegas and Monte Carlo
algorithms hold for a restricted choice of the damping factor $\alpha$ and for $\epsilon \in O(1)$, and do not allow to
choose the target node's score; in contrast, we give (tighter) bounds holding for any $\alpha$, for any possible value
of $\epsilon$ and of the scores (even when non-constant functions of $n$), and under any graph exploration model.
In~\cite{Andersen&2008} Andersen et al.\ give a local algorithm, based on back crawling, to approximate the score contribution of each node in the graph towards a given node $v$ within an additive error $\epsilon_a$.
This algorithm can be used to compute an approximation of $v$'s score, but in the worst case it requires to visit essentially the entire graph -- in line with our $\Omega(n)$ lower bound for ``strictly local'' algorithms.
In~\cite{Bressan&2011} Bressan et al.\ provide, under the same model of~\cite{Bar-Yossef&2008b}, the first formalization
and analysis of the local PageRank ranking problem, which asks to rank a set of nodes in nondecreasing order of PageRank
score, allowing ties for scores within a factor $1+\epsilon$.
Their $\Omega(n)$ bounds for deterministic algorithms and $\Omega(\sqrt{n})$ bounds for Las Vegas and Monte Carlo
algorithms match those of~\cite{Bar-Yossef&2008b}
for the local score approximation problem, leaving open the question if computing the ranking is as hard as computing
the score; and, similarly to~\cite{Bar-Yossef&2008b}, they hold for $\epsilon \in O(1)$ and do not allow choosing scores
arbitrarily.
Finally, Borgs et al.~\cite{Borgs&2012}, in a model that allows random jumps, provide a Monte Carlo algorithm that returns all the nodes with score at least $\Delta$ but no node with score below $\Delta/c$, for any given $c > 3$, by visiting only $\tilde{O}(1/\Delta)$ nodes overall.

\paragraph{Organization of the Paper.}
\label{par:INTRO_organization}
After a short review of PageRank (Section~\ref{sec:pagerank}), Section~\ref{sec:localrank} formally introduces graph
exploration models and formally defines the problem at hand. Section~\ref{sec:det_and_lv} introduces ranking subgraphs,
uses them to prove lower bounds for deterministic and Las Vegas local ranking algorithms under any exploration model,
and analyses the computational complexity of determining their size.
Section~\ref{sec:montecarlo} gives lower bounds for Monte Carlo algorithms under different exploration models, and
provides an algorithm performing almost optimally under every model even if exploiting only a minimum degree of
``non-locality''.

\section{PageRank}
\label{sec:pagerank}
Let $G$ be an $n$-node graph with no dangling nodes (i.e.\ nodes with no outgoing arcs; see the note below). The
PageRank score~\cite{LM06} of a node $v$ is defined as:
\begin{align}
  \label{eqn:pagerank_series_influences}
P(v) = \frac{1-\alpha}{n}\sum_{\tau=0}^{+\infty}\alpha^{\tau}\sum_{z \in
G}\operatorname{inf}_{\tau}(z,v)
\end{align}
where the \textit{damping factor} $\alpha$ is some constant in $(0,1)$, and the $\tau$-step influence
$\operatorname{inf}_{\tau}(z,v)$ of $z$ on $v$ is the probability that a ``random surfer'' starting in $z$, and at each
step following an outgoing arc at random, is in $v$ after exactly $\tau$ steps. 
Using $\operatorname{inf}_{\tau}(z,v)$ to define the \textit{contribution} $P(z,v)$ of $z$ to $v$, the score of $v$ can
be rewritten as:
\begin{align}
  \label{eqn:pagerank_series_contribs}
P(v) = \sum_{z \in G}P(z,v) \qquad\; \text{where} \quad\, P(z,v) =
\frac{1-\alpha}{n}\sum_{\tau=0}^{+\infty}\alpha^{\tau}\operatorname{inf}_{\tau} (z , v)
\end{align}
When $G$ contains dangling nodes, the scores given by Equation~(\ref{eqn:pagerank_series_influences}) add up to less
than $1$, and thus do not form a probability distribution; in this case, PageRank is not defined.
This is typically solved by preprocessing the graph, adding to each dangling node $n$ outgoing arcs towards each
node of $G$. One can prove (see~\cite{Boldi&2006}) that the resulting PageRank scores are a rescaled version of those
given by Equation~(\ref{eqn:pagerank_series_influences}) in the original graph, and thus conserve their ratio and the
ranking they induce on the nodes. For this reason, even in the presence of dangling nodes, we will use the scores given
by Equation~(\ref{eqn:pagerank_series_influences}), which are equivalent to (but easier to compute than) the scores
obtained after the preprocessing. 
Finally, it is worth remarking that the score of a dangling node multiplies by $1/(1-\alpha)$ when adding a self-loop to
it (see~\cite{Bressan&2011}).

\section{Graph\!\ Exploration,\!\ Local\!\ Algorithms,\!\ and\!\ Problem\!\ Definition}
\label{sec:localrank}
This section introduces graph exploration models and local algorithms (Subsection~\ref{sub:explor_models}), formalizes
the local PageRank ranking problem (Subsection~\ref{sub:localcomp}), concluding with remarks on the information
available to algorithms (Subsection~\ref{sub:awareness}).

\subsection{Graph Exploration Models}
\label{sub:explor_models}
A \emph{graph exploration model} describes how algorithms can access a graph by defining a set of \emph{queries}, where
each query receives an input (possibly empty) and reveals a (typically tiny) portion of the graph in output.
In the popular \emph{link server} graph exploration model~\cite{Bar-Yossef&2008b}, a query accepts a node in input and
returns in output a list of all the nodes pointing to it and a list of all the nodes pointed by it.
The \emph{jump and crawl} model~\cite{Brautbar&2010} allows a \textit{jump} query that returns a random node of
the graph, and a \textit{crawl} query equivalent to the link server query; in~\cite{Borgs&2012}, the
\textit{crawl} query returns instead a random child of the input node (and this will be our definition of
\textit{crawl}). In the more minimalist \emph{edge probing} model (see e.g.~\cite{Anagnostopoulos&2012}), a query allows
to check for an arc's existence.
An algorithm obeying the exploration model is called a \emph{local algorithm}, and its goal is to evaluate one or more
nodes in terms of some global graph metric using as few queries as possible -- the number of queries is the \emph{cost}
incurred by the algorithm.
Note that, since the number of queries needed to solve a problem should be a measure of its ``exploration complexity'',
a query should disclose no information about the global arc structure of the graph (see the note in~\ref{apx:1}), if not
for a small portion near the queried node. Formally, given a graph $G$, we require that any graph exploration model
satisfy the following property.
\begin{property}
\label{pty:localmodel}
The output of a query is a (possibly empty) connected subgraph of $G$. For \emph{global queries}, the output must not
depend on the arcs of $G$, but can depend on its set of nodes. For \emph{local queries}, it can depend on the arcs
pointing to and from a node $u$ in input.
\end{property}
\noindent Intuitively, while global queries allow to discover remote nodes of the graph, local queries allow to assess
its local structure around a given node. This property captures any ``natural'' graph exploration model, including all
those proposed in literature (and those mentioned above). It will turn out that the model at hand is irrelevant in some
cases but crucial in others, and that interesting models are built on three types of query:
\begin{itemize}
\item $links(u)$ returns all the parents/children of $u$ and their arcs to/from $u$.
\item $crawl(u)$ returns a child of $u$ uniformly at random, or an empty graph if $u$ has no children.
\item $jump()$ returns a node of $G$ uniformly at random.
\end{itemize}
To capture the hardness of generating unknown, but valid IDs from scratch (think of e.g.\ web URLs), we do not allow
algorithms to query undiscovered nodes, which would require ``guessing'' their ID; instead, nodes should first be
discovered via the queries provided by the model.

\subsection{Local Computation of PageRank}
\label{sub:localcomp}
Given a graph $G$, a damping factor $\alpha$, a set of \emph{target nodes} $v_1, \ldots, v_k \in G$, and a separation
$\epsilon > 0$, the \emph{local PageRank ranking problem} asks to rank the target nodes in nonincreasing order of their
PageRank scores using local algorithms.
Any ranking is valid for nodes whose score ratio is below $1 + \epsilon$, in which case we say there is a tie; otherwise
we say that the nodes are $\epsilon$-separated, and require that the ranking induced by the scores be provided.
The goal of a local (PageRank) ranking algorithm is to compute a correct ranking using as few queries as possible.
It is easy to see that a local algorithm approximating the PageRank scores within a factor $\sqrt{1+\epsilon}$ can
immediately be turned into a local ranking algorithm for separations greater or equal to $\epsilon$ -- therefore, any
lower bound on the cost of local ranking for a separation $\epsilon$ applies to the cost of local
$(\sqrt{1+\epsilon}\,)$-approximation of the score, and the converse is true for upper bounds.
We denote an instance of the problem as $(G,\{v_1,\ldots,v_k\})$, omitting $\epsilon$ and $\alpha$ for easiness of
notation.
We write $u \succ v$ if $u$ ranks higher than $v$.
Special attention will be paid to scenarios where the target nodes are among the top ranking in the graph, which
are often the only ones of interest.

\subsection{Awareness Degree of the Algorithms}
\label{sub:awareness}
The cost incurred by an algorithm may reduce if the algorithm is aware of some global property of the graph (e.g.\ its
size or its diameter). However, assuming any knowledge appears arbitrary, since in many cases only partial information
is available (especially in massive or evolving graphs), and this information also varies from case to case.
For these reasons, and consistently with the (sometimes tacit) assumptions of previous work~\cite{Andersen&2008,
Bahmani&2012}, we will consider obliviousness as the reference case -- the algorithms do not have any kind of awareness
about global properties of the graph, but possess only the information provided by the output of the queries.

\section{Deterministic and Las Vegas}
\label{sec:det_and_lv}
This section investigates the interplay between graph exploration models and local ranking algorithms that return always
a correct ranking (i.e.,\ deterministic and randomized Las Vegas algorithms).
The variety of graph exploration models satisfying Property~\ref{pty:localmodel} makes it non-trivial to analyse these
local ranking algorithms ``in general''. Surprisingly, this can be done by analysing the algorithms under a specific,
strictly local exploration model that allows only \links queries; in this context, the two simple notions of
\textit{visit subgraph} (Subsection~\ref{subsec:visit_subgraphs}) and \textit{ranking subgraph}
(Subsection~\ref{subsec:ranking_subgraphs}) completely characterize the behaviour of any algorithm that returns always a
correct solution, and it will turn out that ranking subgraphs give per-instance lower bounds holding also under any
other exploration model.
Leveraging on these results, we build worst-case instances that yield tight lower bounds holding simultaneously under
any graph exploration model (Subsection~\ref{subsec:correct_cost}). Finally, we prove that, in general, computing or
even approximating these lower bounds is an NP-complete problem (Subsection~\ref{sub:computing}).

\subsection{Visit Subgraphs}
\label{subsec:visit_subgraphs}
At any instant, the nodes discovered by a generic algorithm can be divided into a ``kernel'' set of nodes that have
been queried (i.e.\ given in input to a query), and a ``frontier'' set of nodes that have not been queried but appeared
in some query's output. Under a model that allows only \links queries, the algorithm knows all the arcs between kernel
nodes, all the arcs between kernel nodes and frontier nodes, and none of the arcs between frontier nodes. This can be
formalized in the notion of \textit{visit subgraph compatible with a given graph} (see Figure~\ref{fig:visitgraph}),
where a \textit{visit subgraph} describes the general structure of the information collected via \linksNoSp:
\begin{definition}
A {\em visit subgraph} is a graph $H=(V,A)$ where $V$ is partitioned in a {\em kernel} and a {\em frontier} such that 1)
every frontier node has an arc to or from a kernel node and 2) there are no arcs between two frontier nodes.
\end{definition}
\begin{definition}
A visit subgraph $H$ is {\em compatible} with a graph $G$ if 1) $H$ is a subgraph of $G$, 2) the kernel of $H$
induces the same subgraph in $H$ and in $G$, and 3) the nodes pointing from or to the kernel of $H$, and the relative
arcs, are the same in $H$ and in $G$.
\end{definition}
\noindent Intuitively, a visit subgraph $H$ is compatible with a graph $G$ (we also say that $H$ \textit{is on} $G$, or
that $G$ is compatible with $H$) if some sequence of \links queries on $G$ yields exactly $H$. We will denote by
$\kernel(H)$ the kernel of $H$ and by $\frontier(H)$ its frontier, and we will call \textit{kernel size} of $H$ the
cardinality of $\kernel(H)$.
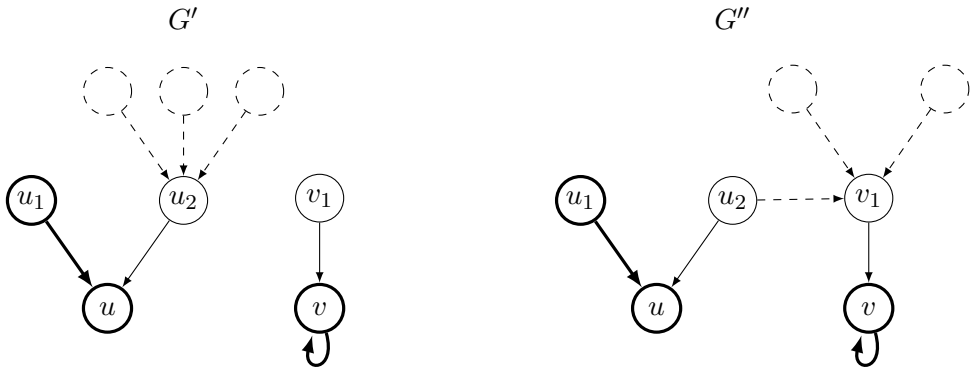
\begin{figure}[h]
\centering
  \begin{tikzpicture}[scale=0.93,
every circle node/.style={inner sep=0pt, minimum size=18},
line/.style={draw},
>=latex,
baseline={(u.base)}]
  \tikzstyle{every node} = [circle, fill=white];
  \def\Gcolor{black};
  \tikzstyle{graph} = [fill=white, draw=\Gcolor, dashed];
  \def\Fcolor{black};
  \tikzstyle{frontier} = [fill=white, draw=\Fcolor];
  \def\Kcolor{black};
  \tikzstyle{kernel} = [fill=white, draw=\Kcolor, very thick];
  
  \node[kernel] (u) at (-2,0) {\color{\Kcolor}$u$};
  \node[above=0.77cm of u] (xu) {};
  \node[kernel] (u1) [left of=xu] {\color{\Kcolor}$u_1$};
  \path[kernel,->,\Kcolor] (u1) edge (u);
  \node[frontier] (u2) [right of=xu] {$u_2$};
  \path[frontier,->,\Fcolor] (u2) edge (u);
  \node[graph,above=0.8cm of u2] (yu) {};
  \path[graph,->,\Gcolor] (yu) edge (u2);
  \node[graph] (u4) [left of=yu] {};
  \path[graph,->,\Gcolor] (u4) edge (u2);
  \node[graph] (u5) [right of=yu] {};
  \path[graph,->,\Gcolor] (u5) edge (u2);

  \node[kernel] (v) at (1,0) {\color{\Kcolor}$v$};
  \path[kernel,->,\Kcolor] (v) edge[loop below] (v);
  \node[frontier, above=0.8cm of v] (v1) {$v_1$};
  \path[->,\Fcolor] (v1) edge (v);

  \node[above=1.77cm of u2] (textG) {$G'$};
 \end{tikzpicture}
 \hspace{2.5cm}
 \begin{tikzpicture}[scale=0.93,
every circle node/.style={inner sep=0pt, minimum size=18},
line/.style={draw},
>=latex,
baseline={(u.base)}]
  \tikzstyle{every node} = [circle, fill=white];
  \def\Gcolor{black};
  \tikzstyle{graph} = [fill=white, draw=\Gcolor, dashed];
  \def\Fcolor{black};
  \tikzstyle{frontier} = [fill=white, draw=\Fcolor];
  \def\Kcolor{black};
  \tikzstyle{kernel} = [fill=white, draw=\Kcolor, very thick];

  \node[kernel] (u) at (-2,0) {\color{\Kcolor}$u$};
  \node[above=0.77cm of u] (xu) {};
  \node[kernel] (u1) [left of=xu] {\color{\Kcolor}$u_1$};
  \path[kernel,->,\Kcolor] (u1) edge (u);
  \node[frontier] (u2) [right of=xu] {$u_2$};
  \path[frontier,->,\Fcolor] (u2) edge (u);

  \node[kernel] (v) at (1,0) {\color{\Kcolor}$v$};
  \path[kernel,->,\Kcolor] (v) edge[loop below] (v);
  \node[frontier, above=0.8cm of v] (v1) {$v_1$};
  \path[frontier,->,\Fcolor] (v1) edge (v);
  \node[above=0.8cm of v1] (yv) {};
  \node[graph] (v2) [left of=yv] {};
  \path[graph,->,\Gcolor] (v2) edge (v1);
  \node[graph] (v3) [right of=yv] {};
  \path[graph,->,\Gcolor] (v3) edge (v1);
  \path[graph,->,\Gcolor] (u2) edge (v1);

  \node[above=1.77cm of u2] (textG) {$G''$};
 \end{tikzpicture}
 \caption{
A visit subgraph $H$ (solid arcs and nodes), with $\kernel(H) = \{u,u_1,v\}$ (in bold) and $\frontier(H) =
\{u_2,v_1\}$, compatible with two different graphs $G'$ and $G''$ (left and right, made of all the nodes and arcs
included those dashed). For e.g.\ $\alpha = 0.3$, this visit subgraph is not a ranking subgraph (see
Subsection~\ref{subsec:ranking_subgraphs})
for $u$ and $v$, since $u \succ v$ in $G'$ but $u \prec v$ in
$G''$.
}
 \label{fig:visitgraph}
\end{figure}

\subsection{Ranking Subgraphs}
\label{subsec:ranking_subgraphs}
Since an algorithm holds only the information given by the visit subgraph retrieved so far (see
Subsection~\ref{sub:awareness}), it cannot decide which of the different compatible input graphs it is exploring.
Nonetheless, that information may still be enough to deduce the correct ranking of the target nodes -- 
in this case, we say that the visit subgraph is a \emph{ranking subgraph}, which can be defined as follows:
\begin{definition}
A visit subgraph $H$ is a {\em ranking subgraph} for the nodes $v_1, \ldots, v_k$ if $v_1, \ldots, v_k \in
\kernel(H)$ and there do not exist two graphs compatible with $H$ where $v_1, \ldots, v_k$ have different relative
rankings (however ties, according to the separation $\epsilon$, are broken).
\end{definition}
\noindent We can now prove that, under a model that allows only \links queries, any algorithm that returns always a
correct ranking must visit always a ranking subgraph:
\begin{theorem}
\label{thm:must_visit_ranking_set}
Let $ALG$ be a deterministic or randomized Las Vegas local ranking algorithm using only \links queries, and let $H$ be
the visit subgraph visited by any execution of $ALG$ on a given instance $(G,\{v_1,\ldots,v_k\})$.
Then $H$ is a ranking subgraph for $v_1,\ldots,v_k$.
\end{theorem}
\noindent By Theorem~\ref{thm:must_visit_ranking_set}, in the strictly local model allowing only \links queries, the
cost incurred by any deterministic or randomized Las Vegas algorithm on the instance $(G,\{v_1,\ldots,v_k\})$ is lower
bounded by the smallest kernel size of all the ranking subgraphs for $v_1,\ldots,v_k$ compatible with $G$. Crucially, we
can prove that this holds also under any other model satisfying Property~\ref{pty:localmodel}, whatever the combination
of local and global queries it allows. Formally:
\begin{theorem}
\label{thm:all_must_visit_ranking_graph}
Consider an instance $(G,\{v_1,\ldots,v_k\})$, and let $r^*$ be the smallest kernel size of all the ranking subgraphs
for $v_1,\ldots,v_k$ on $G$.
Then, under any graph exploration model satisfying Property~\ref{pty:localmodel}, any execution of any deterministic or
randomized Las Vegas local ranking algorithm performs at least $r^*$ queries on the instance $(G,\{v_1,\ldots,v_k\})$.
\end{theorem}
\noindent
By Theorem~\ref{thm:all_must_visit_ranking_graph}, ranking subgraphs give per-instance lower bounds, independent of the
exploration model, on the cost of solving the local ranking problem.
It turns out that ranking subgraphs are characterized by two simple properties, which can be expressed in terms of
\emph{kernel score}:
\begin{definition}
Consider a visit subgraph $H$ and a node $v \in \kernel{(H)}$. The {\em kernel contribution} $P_H(z,v)$ is the
contribution of $z$ to $v$ through paths whose nodes (including $z$) are all in $\kernel{(H)}$. The {\em kernel score}
$P_H(v)$ is:
\begin{align}
  P_H(v) = \sum_{z \in H}P_H(z,v)
\end{align}
\end{definition}
\noindent Thus, $P_H(\cdot)$ is the ``standard'' PageRank score on the graph $H$, except that it disregards any
contribution coming from or through $\frontier{(H)}$.
Therefore, given $H$, $P_H(\cdot)$ lower bounds the PageRank score on any graph compatible with $H$, where this
contribution in general exists; and since it can be arbitrarily small (as frontier nodes can have arbitrarily many
children), the bound is the tightest possible.
We can now use $P_H(\cdot)$ to characterize ranking subgraphs:
\begin{theorem}
\label{thm:rankgraph_char}
A\! visit subgraph $H$\! is a ranking subgraph for two kernel nodes $u \succ\! v$ if and only if:
\begin{enumerate}
 \item $P_H(u) \geq P_H(v)/(1+\epsilon)$
 \item $\forall w \in \frontier(H)$ : $\underset{(w,z)\in H}{\sum} P_H(z,u)
 \geq \underset{(w,z)\in H}{\sum} P_H(z,v)$.
\end{enumerate}
\end{theorem}
\noindent Theorem~\ref{thm:rankgraph_char} implies that an algorithm can use $P_H(\cdot)$ to assess if the current visit
subgraph $H$ is a ranking subgraph and, in that case, immediately return the correct ranking.
Subsection~\ref{subsec:correct_cost} exploits all these results to develop worst-case lower bounds on the cost of local
ranking; the proofs will make use of the notion of union between visit subgraphs (formalized
in~\ref{apx:ranking_subgraphs_lemmas}).

\subsection{Lower Bounds}
\label{subsec:correct_cost}
Leveraging on the concept of ranking subgraph and related results (Subsection~\ref{subsec:ranking_subgraphs}), we are
now able to prove lower bounds for deterministic and randomized Las Vegas algorithms that are independent from the graph
exploration model. Formally:
\begin{theorem}
 \label{thm:cost_deterministic}
Choose an integer $k \geq 2$, a damping factor $\alpha \in (0,1)$, a score function $\Theta(1/x) \leq
p(x) \leq \Theta(1)$, and a separation function $\Theta(1/x) \leq \epsilon(x) \leq \Theta(1/p(x)^\frac{1}{k-1})$.
There exists a family of graphs $\{G_m\}$ where $G_m$ has size $n=\Theta(m)$ and contains $k$ nodes $v_1,\ldots,v_k$
such that:
\begin{enumerate}
 \item $v_1, \ldots, v_k$ have score $\Omega(p(n))$
 \item $v_1, \ldots, v_k$ are $\Theta(\epsilon(n))$-separated, i.e.\ $P(v_{i+1})/P(v_i) \in \Theta(1+\epsilon(n))$
 \item $v_1, \ldots, v_k$ are the top $k$ ranking nodes in $G_m$
 \item under any graph exploration model, any execution of any deterministic or randomized Las Vegas local ranking
algorithm needs $n(1 - O(p(n)\epsilon(n)(1 + \epsilon(n))^{k-2}))$ queries to solve the instance
$(G_m,\{v_1,\ldots,v_k\})$.
\end{enumerate}
\end{theorem}
\noindent Theorem~\ref{thm:cost_deterministic} strengthens and extends to every exploration model the $\Omega(n)$ lower
bound for deterministic algorithms limited to \links queries, and the analogous $\Omega(\sqrt{n})$ lower bound for Las
Vegas algorithms (see~\cite{Bressan&2011}); and it does so also for the problem of computing within a factor
$\Theta(\sqrt{1+\epsilon(n)\!}\,)$ the scores of the target
nodes (whose solution gives the ranking)~\cite{Bar-Yossef&2008b}.
Furthermore, when the absolute score separation $p(n)\epsilon(n)(1 + \epsilon(n))^{k-2}$ between $v_{k-1}$ and $v_k$ is
in $o(1)$ (see the proof in~\ref{apx:proof_thm_cost_deterministic}), the lower bound goes as far as $n - o(n)$, meaning
that the algorithm must explore the entire graph, save possibly a vanishingly small portion.

In the next subsection we will prove that computing or even approximating these general, tight lower bounds on generic
instances is an NP-complete problem.

\subsection{Ranking Subgraphs and Computational Complexity}
\label{sub:computing}
If we were able to compute the minimal kernel size of all the ranking subgraphs for any given instance, we could obtain
lower bounds for any graph and any set of target nodes -- including those deriving from real-world scenarios. 
Unfortunately, as the present subsection shows, this problem is computationally intractable unless P $=$ NP.

Given a graph $G$ and two target nodes $u$ and $v$ (the results can clearly be extended to $k \geq 2$ target nodes), we
denote by RANKGRAPH the problem of deciding if $G$ is compatible with a ranking subgraph of given kernel size $r$ for $u
\succ v$. Formally:
\vspace{1em}\\
RANKGRAPH$(G,u,v,r)$
\par {\bf Input}: directed graph $G$, target nodes $u, v \in G$, integer $r > 0$.
\par {\bf Output}: YES if there exists a ranking subgraph of kernel size $r$ for $u \succ v$ compatible with $G$, NO
otherwise.
\vspace{1em}\\
Theorems~\ref{thm:rankgraph_nph_clique} and~\ref{thm:rankgraph_nph_dominating} prove that RANKGRAPH is
\mbox{NP-hard} even when the complexity deals only with, respectively, Property 1 and Property 2 of
Theorem~\ref{thm:rankgraph_char} (part of the kernel set being already fixed by the other property); thus, the problem
appears to be the ``superposition'' of two \mbox{NP-hard} problems.
\begin{theorem}
\label{thm:rankgraph_nph_clique}
 RANKGRAPH is NP-hard (by reduction from CLIQUE).
\end{theorem}

\begin{theorem}
\label{thm:rankgraph_nph_dominating}
 RANKGRAPH is NP-hard (by reduction from DOMINATING SET).
\end{theorem}
\noindent Furthermore, since DOMINATING SET is NP-hard to approximate within $o(\log(|G|))$~\cite{Raz&1997}, and since
the reduction of Theorem~\ref{thm:rankgraph_nph_dominating} preserves (within small constant factors) the size of both
the input graph and the dominating set, we have:
\begin{lemma}
\label{lem:inapproximability}
Approximating RANKGRAPH within $o(\log(|G|))$ is NP-hard.
\end{lemma}
\noindent Note that verifying if a given visit subgraph is a ranking subgraph for $u \succ v$ in $G$ takes polynomial
time -- one need only check that its kernel set contains $u$ and $v$, and that it satisfies both properties of
Theorem~\ref{thm:rankgraph_char} -- and hence RANKGRAPH is NP-complete, even accepting an approximation factor 
$o(\log(|G|))$.

\section{Monte Carlo Algorithms}
\label{sec:montecarlo}
This section investigates the interplay between graph exploration models and Monte Carlo local ranking algorithms. It
turns out that the choice of the exploration model is crucial: depending on the type of queries
available, an algorithm can perform as poorly as possible or instead very efficiently.
We start by showing that, differently from deterministic and Las Vegas algorithms, Monte Carlo local ranking algorithms
cannot be characterized using the theoretical machinery of ranking subgraphs (Subsection~\ref{sub:MC_avoid_rankgraph});
and nonetheless, under any ``strictly local'' graph exploration model, they still undergo essentially the
same lower bounds (Subsection~\ref{sub:MC_bounds_nojump}).
On the other hand, with just a ``minimum amount'' of non-locality, we are able to build an efficient, random walk-based
Monte Carlo local ranking algorithm (Subsection~\ref{sub:MC_samplerank}), which we pair with almost matching lower
bounds that hold under all graph exploration models (Subsection~\ref{sub:MC_bounds_jump}).

\subsection{Eluding Ranking Subgraphs}
\label{sub:MC_avoid_rankgraph}
Since ranking subgraphs consist of portions of the graph an algorithm must visit to guarantee always a correct
result, one may think that Monte Carlo algorithms can in part avoid them. Indeed, it turns out that Monte 
Carlo algorithms can elude ranking subgraphs in a strong sense: we show that, even under an exploration model allowing
only local queries, there exists one that always solves correctly a particular instance and yet never visits any of its
ranking subgraphs, while guaranteeing an arbitrarily small error probability on every other instance. Formally: 
\begin{theorem}
\label{thm:avoid_ranking_graphs}
For any $\eta > 0$ there exist a \linksNoSp-based Monte Carlo local ranking algorithm with confidence $1-\eta$ and an
instance $(G,\{u,v\})$ such that the algorithm always solves correctly $(G,\{u,v\})$ without visiting any of its ranking
subgraphs.
\end{theorem}\noindent
At this point, one might wonder if Monte Carlo algorithms can outperform deterministic and Las Vegas algorithms.
As long as the exploration model allows only local queries, however, essentially the same bounds hold, and any Monte
Carlo algorithm with non-trivial confidence may have to visit almost the entire input graph, as the next subsection
shows.

\subsection{The Limited Power of Local Information}
\label{sub:MC_bounds_nojump}
We prove that, under any graph exploration model that does not allow global queries (thus forcing the exploration to
expand from the target nodes via local queries), in the worst case one must visit essentially the entire input graph
even to perform just slightly better than returning a ranking at random. Formally:
\begin{theorem}
 \label{thm:cost_monte_carlo_onlycrawl}
Choose an integer $k \geq 2$, a damping factor $\alpha \in (0,1)$, a separation
$\epsilon > 0$, and a score function $\Theta(1/x) \leq p(x) \leq \Theta(1)$.
Under any graph exploration model allowing only local queries, for any
Monte Carlo local ranking algorithm MC with confidence $\frac{1}{k!} +
\delta$ there exists a family of graphs $\{G_m\}$ where $G_m$ has size
$n=\Theta(m)$ and contains $k$ nodes $v_1,\ldots,v_k$ such that:
\begin{enumerate}
 \item $P(v_1), \ldots, P(v_k) \in \Theta(p(n))$
 \item $v_1, \ldots, v_k$ are $\approx$ $\epsilon$-separated, i.e.\
$P(v_{i+1})/P(v_i) \approx 1+\epsilon$
 \item MC performs $\Omega(\delta n)$ queries on the instance
$(G,\{v_1,\ldots,v_k\})$.
\end{enumerate}
\end{theorem}
\noindent By Theorem~\ref{thm:cost_monte_carlo_onlycrawl}, as long as we are confined to local queries, the simultaneous
presence of randomization and of a positive and possibly large probability of error yields almost no advantage over
deterministic algorithms (that cannot perform random choices and can never return an incorrect ranking).
As the next subsection shows, this is the borderline: adding just a ``minimum amount'' of non-locality allows to build
much more efficient algorithms, even if the local queries available are far less informative than e.g.\ \linksNoSp.

\subsection{SampleRank}
\label{sub:MC_samplerank}
Under a model allowing only \jump and \crawl queries, we build an efficient, random walk-based Monte Carlo algorithm
for the local computation of PageRank ranking, with guarantees on both the confidence and the cost.
The algorithm stems from existing random walk-based techniques for the local approximation of scores, and estimates
the ranking behaving optimally 
(see the bounds of~\cite{Borgs&2012}) for every possible value of the nodes' scores and separation
(unlike e.g.\ that of~\cite{Borgs&2012}, which is designed to work only for $\epsilon > 2$).

Our algorithm is based on the following node sampling routine, that emulates a random walk using \jump and \crawl
queries. As already proven in~\cite{Avrachenkov&2007}, the routine returns node $v$ with probability exactly $P(v)$;
Lemma~\ref{lem:sample} gives a proof which is more insightful in this context.
\begin{algorithm}
\label{alg:photo}\caption{SampleNode$()$}
\begin{algorithmic}
\State current\_\,node $\leftarrow$ \jump
\Loop
\State with probability $(1-\alpha)$ \textbf{return} current\_\,node
\State current\_\,node $\leftarrow$ \textit{crawl}(current\_\,node)
\NoEndIf{current\_\,node $= \emptyset$}
  current\_\,node $\leftarrow$ \jump
\EndLoop
\end{algorithmic}
\end{algorithm}\vspace{-1.5ex}
\begin{lemma} 
\label{lem:sample}
$SampleNode()$ returns node $v$ with probability equal to $P(v)$.
\end{lemma}
\noindent By a simple Chernoff bound, $m$ calls to SampleNode perform approximately $m/(1-\alpha)$ queries (this could
be reduced to $m$, but making the algorithm more complicated to analyse~\cite{Avrachenkov&2007}):
\begin{lemma} 
\label{lem:sampletime}
One call to SampleNode performs in expectation less than $\frac{2}{1-\alpha}$ queries. The probability that $m$ calls to
SampleNode perform more than $\frac{2(1+\Delta)m}{1-\alpha}$ queries is less than
$e^{-\frac{m}{2}\cdot\frac{\Delta^2}{1+\Delta}}$.
\end{lemma}
\noindent By Lemma~\ref{lem:sample}, we can estimate the PageRank score of a node $v$ as the fraction $\hat P_m(v)$ of
$m$ calls to SampleNode that returned $v$, and we can in turn estimate the relative ranking of a set of nodes by
repeatedly calling SampleNode until their estimated scores ``stabilize''.
This intuition is formalized in our algorithm SampleRank, which repeatedly samples the nodes of $G$ until, according to
confidence intervals based on $\hat P_m$, each pair of target nodes either appears to be $\epsilon$-separated or appears
to yield a tie.
\begin{algorithm}
\caption{SampleRank$(G,v_1,\ldots,v_k,1-\eta)$}
\label{alg:photorank2}
\begin{algorithmic}
\Repeat
\State perform the $m$-th call to SampleNode
\State $\hat P_m(v_i) \leftarrow$ fraction of SampleNode calls that returned $v_i$
\State $C_m(v_i) = ( P^L_m(v_i), P^U_m(v_i) ) \leftarrow$ $(1-\eta/k)$-confidence interval for $P(v_i)$ based on $\hat
P_m(v_i)$
\Until{$\forall i \neq j :$ $\big(C_m(v_i) \cap C_m(v_j) = \emptyset \big) \vee \big(P^U_m(v_i)/ P^L_m(v_j) \leq 1 +
\epsilon\big) \vee \big(P^U_m(v_j)/ P^L_m(v_i) \leq 1 + \epsilon\big)$}
\State \Return the ranking of $v_1,\ldots,v_k$ induced by $\hat P_m$
\end{algorithmic}
\end{algorithm}
\newline
We prove that:
\begin{theorem}
\label{thm:samplerank2}
Consider an instance $(G,\{v_1,\ldots,v_k\})$, and let $p = \min_{i=1,\ldots,k}\{P(v_i)\}$.
A call to SampleRank\,$(G,v_1,\ldots,v_k, 1 - \eta)$ has probability at least $1-\eta$ of returning a correct ranking of
$v_1,\ldots,v_k$ while performing  $O\big(\frac{1}{1 - \alpha}\log(\frac{4k}{\eta})\frac{(1+\epsilon)^2}{\epsilon^2
p}\big)$
queries.
\end{theorem}
\noindent
For fixed $\alpha$, $\epsilon$, $k$ and $\eta$, SampleRank performs essentially $O(1/p)$ queries, which ranges from
$O(1)$ to $O(n)$ as a function of $p$, and for $p \in \omega(1/n)$ is asymptotically less than the $\Omega(n)$ bound
holding under strictly local exploration models (see Subsection~\ref{sub:MC_bounds_nojump}).
And note that SampleRank relies only on \jump and \crawl queries, which are in a sense the sparest possible global
query (returning a random node of the graph), and the sparest possible local query (returning a random child of
the queried node).
This provides the first theoretical justification as to why all the existing local algorithms for PageRank score or rank
estimation either assume to have global access to the graph~\cite{Bahmani&2010,Bahmani&2012,Borgs&2012}, or fail to give
adequate worst-case cost guarantees~\cite{Andersen&2008,Bar-Yossef&2008b,Bressan&2011,CGS04}.

The next subsection will show that, even under any other exploration model, no algorithm can perform better than
SampleRank, at least as long as the scores do not fall under $\Theta(n^{-2/3})$.

\subsection{Lower Bounds for All Graph Exploration Models}
\label{sub:MC_bounds_jump}
This subsection presents lower bounds for Monte Carlo local algorithms that hold under every graph exploration model
allowing any combination of local and global queries.
These are the first lower bounds at this level of generality that we are aware of.
Formally, we prove that:
\begin{theorem}
 \label{thm:cost_monte_carlo}
Choose an integer $k \geq 2$, a damping factor $\alpha \in (0,1)$, a separation
$\epsilon > 0$, and a score function $\Theta(1/x) \leq p(x) \leq \Theta(1)$.
Under any graph exploration model, for any Monte Carlo local ranking
algorithm MC with confidence $\frac{1}{k!} + \delta$ there exists a family of
graphs $\{G_m\}$ where $G_m$ has size $n=\Theta(m)$ and contains $k$ nodes
$v_1,\ldots,v_k$ such that:
 \begin{enumerate}
  \item $P(v_1), \ldots, P(v_k) \in \Theta(p(n))$
  \item $v_1, \ldots, v_k$ are $\approx$ $\epsilon$-separated, i.e.\ $P(v_{i+1})/P(v_i) \approx 1+\epsilon$
  \item MC performs $\Omega( \delta \operatorname{min}\{1/p(n), n^{2/3}\})$ queries on the instance
$(G,\{v_1,\ldots,v_k\})$.
 \end{enumerate}
\end{theorem}
\noindent
By Theorem~\ref{thm:cost_monte_carlo}, our algorithm SampleRank, although based on the spare \jump and \crawl queries,
is asymptotically optimal also among the class of algorithms employing far more informative ones, save possibly for
PageRank scores in $o(n^{-2/3})$.
And, thus, the jump and crawl model appears to be the most powerful possible for local PageRank computations, except
perhaps for extremely small scores: in this latter case it remains to decide whether the discrepancy is an artifact of
our analysis, or if more informative global and local queries allow to obtain better (and fully sublinear?)\ worst-case
cost guarantees.
\section{Conclusions}
\label{sec:conclusions}
In the colourful universe of results and techniques for local PageRank computations, we an\-alyse the interplay
between algorithms and graph exploration models, identifying the boundaries be\-tween reasonable and prohibitive
exploration costs.
We show that one must combine a non-local ex\-plora\-tion and a positive probability of error to avoid exploring
essentially the whole graph in the worst case; and that if both these conditions are satisfied, a spare model
providing minimum amounts of local and global access is as powerful as possible, except perhaps for extremely small
PageRank scores, and allows to obtain the solution by visiting a limited portion of the graph.

For the practitioner, our results imply that it is impossible to ensure correctness when com\-put\-ing the PageRank
score or  ranking in graphs that (as the web or many social networks) are par\-tially inaccessible or evolve rapidly
in time, even if allowing for large approximation factors;
and they suggest great care when designing access models (e.g.\ in graph manipulation libraries or social networks
APIs) that should support local PageRank computations -- either you include a global exploration primitive, or you 
cut out anyone who attempts efficient local computations.

For the theoretician, our work delineates the general framework of local PageRank computation, leaving open the only
problem to determine if even the smallest PageRank scores can be computed locally for a sublinear cost, or if this
apparent possibility is an artifact of our lower bounds.
If the first case holds, we have shown that the only possibility to achieve a sublinear cost resides in using local and
global exploration primitives more informative than the ones made available by the simple jump and crawl model.

\clearpage
\appendix

\section{Appendix}

\subsection{A note for Subsection~\ref{sub:explor_models}}
\label{apx:1}
If a query can return an output that depends on the global arc structure of the graph, then it may serve as an oracle to
solve global problems for a minimum cost -- e.g.,\ it could code the PageRank score of an input node into a properly
crafted sequence of output nodes. To exclude this possibility, all the ``reasonable'' graph exploration models allow,
often implicitly, only queries that do not disclose information about the global arc structure of the graph, if not for
a small portion near the node given in input to the query itself.

\subsection{Proof of Theorem~\ref{thm:must_visit_ranking_set}}
\begin{proof}
At any point during the exploration, $ALG$ does not possess any information
besides that provided by the queries' output, and therefore besides that
contained in $H$ (which contains all the information collected along the whole
execution). Therefore, $ALG$ does not discriminate between different input
graphs whose exploration can yield exactly $H$ -- i.e.\ between different graphs
that are compatible with $H$. Hence $ALG$ behaves identically on any graph
compatible with $H$, querying the same nodes and returning the same result (in
the case of Las Vegas randomized algorithms, this must happen in at least some
execution). But if $H$ is not a ranking subgraph for $v_1,\ldots,v_k$, then
there are two graphs $G'$ and $G''$, both compatible with $H$, where
$v_1,\ldots,v_k$ have different rankings; and therefore the ranking returned by
$ALG$ must be wrong for at least one of the two instances
$(G',\{v_1,\ldots,v_k\})$ and $(G'',\{v_1,\ldots,v_k\})$. Thus $ALG$ does not
return always a correct ranking, contradicting the hypotheses.
\end{proof}

\subsection{Proof of Theorem~\ref{thm:all_must_visit_ranking_graph}}
\begin{proof}
Suppose by contradiction that there exists an algorithm $ALG$ that obeys a model
satisfying Property~\ref{pty:localmodel} and that solves (at least in some
execution) the instance $(G,\{v_1,\ldots,v_k\})$ performing $r < r^*$
queries. Consider the set formed by any node that either has been returned by a
global query issued by $ALG$, or either is that on which the output of a local
query issued by $ALG$ depends (see Property~\ref{pty:localmodel}); this set has
then size at most $r$. Consider the ranking subgraph compatible with $G$ that
has  this set has kernel set. Clearly, $H$ is a supergraph of the graph
witnessed by $ALG$, and its kernel nodes are known to have no neighbours outside
the visit subgraph itself -- thus it contains at least all the
information possessed by $ALG$. All the input graphs compatible with $H$ must
then be ``compatible'' also with the visit performed by $ALG$, i.e.,\ the
queries' output witnessed by $ALG$ can derive from any of them with positive
probability. $H$ is not a ranking subgraph for $v_1,\ldots,v_k$ on $G$, since
its kernel size is smaller than $r^*$; therefore it is compatible with two
graphs where the target nodes have different rankings. These two graphs
must then be compatible also with the visit performed by $ALG$. By the same
argument of the proof of Theorem~\ref{thm:must_visit_ranking_set}, $ALG$ does
not always return a correct solution, leading to a contradiction.
\end{proof}

\subsection{Proof of Theorem~\ref{thm:rankgraph_char}}
\begin{proof}
We first show that the two conditions are both necessary.

For 1), we show that if $P_H(u) < P_H(v)/(1+\epsilon)$ then there is a graph $G$
compatible with $H$ where $P(v) > (1+\epsilon)P(u)$, hence $v \succ u$ and $H$
is not a ranking subgraph for $u \succ v$. $G$ is built from $H$ by adding $n$
nodes that are pointed by every frontier node of $H$ (if the frontier is empty,
then these nodes are orphans).
The scores $P(u)$ and $P(v)$ of $u$ and $v$ in $G$ are the sum of two
components: one accounting for paths that contain only kernel nodes, the other
accounting for paths that include at least one frontier node. The first
component amounts to $P_H(\cdot)|H|/|G|$, since paths containing only kernel
nodes are the same in $H$ and $G$, but in $G$ the probability of restarting at
any given node is $1/|G|$ instead of $1/|H|$. The second component is the sum of
contributions of at most $|H|$ different nodes, since this is the number of
ancestors of $u$ or $v$ in $H$ and, by construction, also in $G$; and these
contributions are provided via paths each containing some frontier node $x$,
which has a fraction at most $|H|/(|H| + n) = |H|/|G|$ of children leading to
the target node (the others leading to the additional $n$ sinks). Taking into
account the normalization factor $(1 - \alpha)/|G|$, the second component is
bounded by $|H| \cdot (1 - \alpha)/|G| \cdot |H|/|G| < |H|^2/|G|^2$. Therefore
we have:
\begin{align}
 \frac{P(u)}{P(v)} &\leq \frac{P_H(u)|H|/|G| + |H|^2/|G|^2}{P_H(v)|H|/|G|} =
\frac{P_H(u)}{P_H(v)} + \frac{|H|}{P_H(v)|G|}
\end{align}
but by hypothesis $P_H(u)/P_H(v) < 1+\epsilon$, and thus choosing a sufficiently
large $n$ (and therefore $|G|$) will make $P(v) > (1+\epsilon)P(u)$, implying
that $v \succ v$ in $G$ and that $H$ is not a ranking subgraph for $u \succ v$.

For 2), if there exists $w \in \frontier{(H)}$ such that
$\sum_{(w,z)\in H} P_H(z,u) < \sum_{(w,z)\in H} P_H(z,v)$, then again there is a
graph $G$ compatible with $H$ where $P(v) > (1+\epsilon)P(u)$, hence $v \succ u$
and $H$ is not a ranking subgraph for $u \succ v$. $G$ is obtained from $H$ by
adding $n$ parents to $w$. In $G$, we have $P(w,v) > P(w,u)$ since by hypothesis
the children of $w$ provide, on average, strictly more contribution $v$ than to
$u$. Therefore, each of the additional parents of $w$ also provides strictly
more contribution to $v$ than to $u$. Then for $n$ sufficiently large we have
again $P(v) > (1+\epsilon)P(u)$, implying that $v \succ u$ in $G$ and that $H$
is not a ranking subgraph for $u \succ v$.

To prove that the two conditions are sufficient, consider any graph $G$
compatible with $H$ and consider again the two components of the contributions
provided via paths that contain respectively only kernel nodes and at least one
frontier node. As discussed above, the paths related to the first component
still exist in $G$ and provide exactly the same contributions, except for a
scaling factor $|H|/|G|$; therefore, by condition 1) their overall contribution
towards $u$ is at least $1/(1+\epsilon)$ times that towards $v$. The paths
related to the second component must contain a frontier node $w$, and condition
2) ensures that, of the contributions flowing through $w$, the component
directed to $u$ is greater or equal than (and thus at least $1/(1+\epsilon)$
times) that directed to $v$. Summing the two components proves that in $G$ we
still have $P(u) \geq P(v)/(1+\epsilon)$, and therefore either we have a tie,
which we can break as $u \succ v$, or we have $P(u) > (1+\epsilon)P(v)$, which
already gives $u \succ v$. In any case, the ranking is $u \succ v$ and $H$ is a
ranking subgraph.
\end{proof}

\subsection{Union of visit subgraphs}
\label{apx:ranking_subgraphs_lemmas}
\begin{definition}
Let $H_1 = (V_1,A_1)$ and $H_2 = (V_2,A_2)$ be two visit subgraphs.
Their {\em union} $H_1 \cup H_2$ is $H = (V,A)$ such that $V = V_1 \cup V_2$, $A = A_1 \cup A_2$, and
$\operatorname{kernel}(H) = \operatorname{kernel}(H_1) \cup \operatorname{kernel}(H_2)$.
\end{definition}
\begin{lemma}
 \label{lemma:visitgraph_union}
Let $H = H_1 \cup H_2$ be the union of two visit subgraphs $H_1$ and $H_2$. Then
\begin{enumerate}
 \item $H$ is a visit subgraph
 \item if $H_1$ and $H_2$ are compatible with a subgraph $G$, then $H$ is also compatible with $G$
 \item if $H_1$ is a ranking subgraph for two nodes $u, v$, then $H$ is also a ranking subgraph for $u, v$
\end{enumerate}
\end{lemma}
\begin{proof}
\begin{enumerate}
 \item If $x \in \frontier(H)$, then $x \in \frontier(H_1) \cup \frontier(H_2)$
and $x \notin \kernel(H_1) \cup \kernel(H_2)$.
From the first follows that $s$ has at least one arc from/to $\kernel(H_1)$ or
$\kernel(H_2)$ and therefore from/to $\kernel(H)$. From the second follows that
$x$ has no arcs from/to $\frontier(H_1)$ or $\frontier(H_2)$, and since the
$\frontier(H) \subseteq \frontier(H_1) \cup \frontier(H_2)$, $x$ has no arcs
from/to $\frontier(H)$.
 \item The union $H$ of two subgraphs $H_1$ and $H_2$ of $G$ is clearly a
subgraph of $G$. Consider $x, y \in \kernel(H) = \kernel(H_1) \cup
\kernel(H_2)$; by definition, if and only if $(x,y)$ is an arc of $G$ then it is
also an arc in $H_1$ and $H_2$ (although $x$ or $y$ may be frontier nodes in
one of the two) and therefore in $H$. When $x \in \kernel(H)$ and
$y \in \frontier(H)$, the same reasoning proves that the nodes pointing from/to
the kernel of $H$ in $G$, and the relative arcs, are the same in $H$ and in $G$.
 \item Suppose instead that there exist two graphs $G'$ and $G''$
compatible with $H$ where $u$ and $v$ have different rankings. It is easy to see
that $G'$ and $G''$ are compatible also with $H_1$, which therefore would not be
a ranking subgraph for $u$ and $v$, contradicting the hypothesis.
\end{enumerate}
\end{proof}
\begin{lemma}
 \label{lemma:rankgraph_union}
Let $v_1 \succ \ldots \succ v_k$ be nodes of a graph $G$.
Then the set of all ranking subgraphs for $v_1, \ldots, v_k$ compatible with $G$ coincides with the set of all possible
unions $\bigcup_{i=1}^{k-1}{R_i}$, where $R_i$ is a ranking subgraph for $v_i,v_{i+1}$ compatible with $G$.
\end{lemma}
\begin{proof}
Let $R$ be a ranking subgraph for $v_1, \ldots v_k$ compatible with $G$. Then
$R$ is also a ranking subgraph for $v_i,v_{i+1}$ $\forall i = 1, \ldots, k-1$,
and we can express $R$ as $\bigcup_{i=1}^{k-1}{R_i}$ with $R_i = R$.

On the other side, let $S = \bigcup_{i=1}^{k-1}{R_i}$ where $R_i$ is a ranking
subgraph for $v_i,v_{i+1}$ compatible with $G$. By
Lemma~\ref{lemma:visitgraph_union}, $S$ is compatible with $G$ and is still a
ranking subgraph for $v_i,v_{i+1}$ $\forall i = 1, \ldots, k-1$. Therefore, it
is a
ranking subgraph also for $v_1, \ldots, v_k$.
\end{proof}

\subsection{Proof of Theorem~\ref{thm:cost_deterministic}}
\label{apx:proof_thm_cost_deterministic}
\begin{proof}
We build a graph $G_m$ of size $n = \Theta(m)$ satisfying the statement for
every $m$ sufficiently large. $G_m$ is formed by $k$ disjoint subgraphs $G_m^1,
\ldots, G_m^k$. Subgraph $G_m^i$ (Figure~\ref{fig:cost_deterministic}) contains:
\begin{enumerate}
 \item the target node $v_i$ and its self-loop
 \item $n_i = m \; p(m)  (1+\epsilon(m))^{i-1}$ ``bulk parents'' (rounded to
the nearest positive integer) having $v_i$ as their sole child; if $n_k =
n_{k-1}$, then $v_k$ has $n_k + 1$ bulk parents
 \item for $i < k$, a clique of $m$ nodes, one of which has also an arc to $v_i$
 \item for $1 < i < k$, if $n_i = n_{i-1}$, a ``fractional parent'' of $v_i$
having an arc to $v_i$ and $2 \leq q_i \leq m - 1$ arcs (we compute $q_i$ below)
to clique nodes that do not point to $v_i$
\end{enumerate}
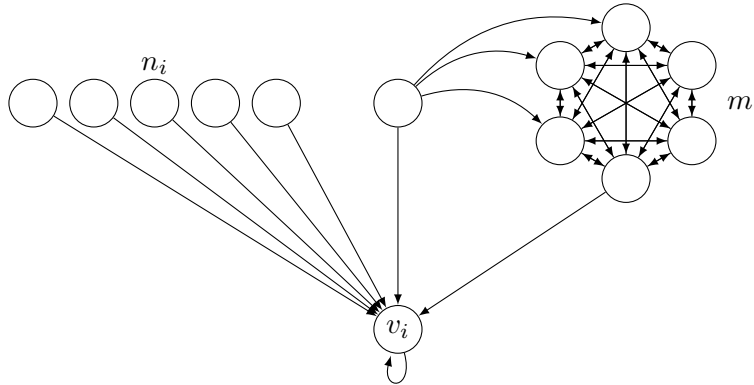
\begin{figure}[h]
 \centering
 \begin{tikzpicture}[scale=1, every circle
node/.style={draw, thin, inner sep=0pt},
line/.style={draw, very thin},
>=latex,
baseline={(vi.base)}]
\tikzstyle{ancestor} = [circle, draw, fill=white, minimum size=18];
\tikzstyle{target} = [circle, draw, fill=white, minimum size=18];

\def\spacing{0.8};
\def\ybasecoord{0};
\def\xbasecoord{4};
\def\ancbasecoord{3};

\node[target] (vi) at (\xbasecoord, \ybasecoord) {$v_i$};

\path[->] (vi) edge[loop below] (vi);
\def\nanc{5};
\pgfmathparse{\nanc/2 + 1/2 - 1};
\let\hoff\pgfmathresult
\foreach \j in {1,...,\nanc} {
  \pgfmathparse{(\j - \hoff) * \spacing};
  \let\xcoord\pgfmathresult;
  \node[ancestor] (vij) at (\xcoord, \ancbasecoord) {};
  \path[->] (vij) edge (vi);
}
\pgfmathparse{(\nanc/2 + 1/2 - \hoff) * \spacing};
\node at (\pgfmathresult, \ancbasecoord + 0.5) {$n_i$};

\def\n{6};
\def\cliqueradius{1};
\node (cliquecenter) at (\xbasecoord + 3, \ancbasecoord) {};

\pgfmathparse{\n - 1};
\let\m\pgfmathresult;
\foreach \x in {0, ..., \m} {
  \pgfmathparse{90 + (\x * 360) / \n};
  \let\angle\pgfmathresult;
      \pgfmathparse{int(\x + 1)};
      \let\xx\pgfmathresult
  \path (cliquecenter) ++(\angle : \cliqueradius) node[ancestor] (u\x) {};
}

\foreach \x in {0, ..., \m} {
  \foreach \y in {1, ..., \m} {
    \pgfmathparse{int(mod(\x + \y, \n))};
    \let\dest\pgfmathresult;
    \draw[->] (u\x) edge (u\dest);
  }
}

\path (cliquecenter) ++(\cliqueradius * 1.5,0) node (cliquelabel) {$m$};

\draw[->] (u3) edge (vi);

\node[ancestor] (vf) at (\xbasecoord, \ancbasecoord) {};
\path[->] (vf) edge (vi);
\foreach \x in {0, ..., 2} {
  \draw[->] (vf) edge[bend left] (u\x);
}
\end{tikzpicture}
\caption{Subgraph $G_m^i$ of graph $G_m$ (Theorem~\ref{thm:cost_deterministic})
for $i < k$. The $n_i$ ``bulk parents'' on the left provide the bulk of $v_i$'s
score, up to a resolution of $\Theta(1/n)$; if still $n_i = n_{i-1}$, then the
``fractional parent'' above $v_i$ provides a score separation up to a resolution
of $\Theta(1/n^2)$. The clique on the right provides a negligible contribution
$O(1/n^2)$, but its $m = \Theta(n)$ nodes must still be visited.}
 \label{fig:cost_deterministic}
\end{figure}

Intuitively, $v_i$ and its $n_i$ bulk parents provide the bulk of $v_i$'s
score, guaranteeing an absolute score separation up to a ``resolution''
$\Omega(1/n)$. If the desired absolute separation is smaller, then $n_i =
n_{i-1}$, in which case the fractional parent provides an additional
contribution $\Theta(1/{n q_i})$, thus guaranteeing an absolute score separation
with resolution up to $\Theta(1/n^2)$. The clique provides a negligible
contribution $O(1/n^2)$ while damping the contribution of the fractional parent
and guaranteeing an $\Omega(n)$ lower bound.

Disregarding the at most $2k$ target nodes and fractional parents, $G_m$ has
size:
\begin{align}
 n \approx \sum_{i = 1}^k (n_i + m) = km + m p(m) \sum_{i = 1}^k (1 +
\epsilon(m))^{i-1}
\end{align}
since $\epsilon(m) \in O(1/p(m)^\frac{1}{k-1})$, this is in $O(1/p(m))$, proving
that $n \in \Theta(m)$. We thus use $n$ instead of $m$ from now on, all the
results still holding asymptotically.

The PageRank score of $v_i$ is the sum of three components. The first, provided
by $v_i$ and its $n_i$ bulk parents, amounts to $\Theta(\frac{1}{n}\alpha
n_i) = \Theta(p(n)(1+\epsilon(n))^{i-1})$. The second, provided by the clique
for $i < k$, can be bounded observing that in an isolated clique (without
incoming or outgoing arcs) each node has score exactly $1/n$, and that adding an
outgoing arc does not increase this score; then expressing $P(v_i)$ as the
weighted sum of the scores of its parents divided by their outdegrees shows that
the whole clique contribution is roughly the score of one of its nodes divided
by $m$, and is thus bounded by $1/nm \in O(1/n^2) = o(p(n))$. The third
component is provided by the fractional parent when $n_i = n_{i-1}$ with $1 < i
< k$; of this, the part flowing through the clique is damped by a factor
$O(1/n^2)$ and is thus negligible, while the part flowing through the single arc
to $v_i$ amounts to $O(1/(n q_i)) = O(p(n))$. It follows that $P(v_1) \in
\Theta(p(n))$ and $P(v_i) \in \Omega(p(n))$ for every $i=1,\ldots,k$.

We prove that $v_1,\ldots,v_k$ are $\Theta(\epsilon(n))$-separated. Consider
$v_i$ and $v_{i+1}$. If $n_i < n_{i+1}$, then the absolute score separation
between $v_{i+1}$ and $v_i$ is guaranteed by their bulk parents since
$n_{i+1}/n_i \approx 1 + \epsilon(n)$ and thus $P(v_{i+1})/P(v_i) \approx 1 +
\epsilon(n)$ -- the clique contribution is negligible and the fractional
parent contribution is at most roughly half that of one bulk parent. If
$n_{i+1} = n_i$, the absolute score separation is between $\Omega(1/n^2)$ and
$O(1/n)$. Then choose an appropriate $2 \leq q_{i+1} \leq \Theta(n)$ to
approximate the separation up to a small constant factor (note that the
contribution flowing to $v_{i+1}$ through the clique is always in $O(1/n^2)$ and
thus dominated by the contribution flowing through the arc pointing directly to
$v_{i+1}$). In any case, the target nodes are $\Theta(\epsilon(n))$-separated.

Note that $v_1,\ldots,v_k$ are the top $k$ ranking nodes in $G_m$, as they
receive the contribution of at least one bulk parent plus that of the fractional
parent. The cliques contain the only other non-orphan nodes, which receive a
contribution of at most $1/n$ from the clique itself (less than a bulk parent)
and from the fractional parent at most the same contribution received by the
corresponding target node.

We conclude proving the lower bound. Any ranking subgraph for $v_1,\ldots,v_k$
in $G_m$ must contain in its kernel set all the nodes of $G_m^1, \ldots,
G_m^{k-1}$, otherwise violating the definition of ranking subgraph (if target
nodes are not in the kernel set) or condition 2 of
Theorem~\ref{thm:rankgraph_char}. It must also contain
$\Omega(n_{k-1}/(1+\epsilon(n)))$ ancestors of $v_k$, otherwise violating
condition 1 of Theorem~\ref{thm:rankgraph_char}. Therefore, at most $O(n_k -
n_{k-1}/(1+\epsilon(n))) = O(n_k \epsilon(n)/(1+\epsilon(n)))$ nodes of $G_m$
may not be in the kernel set. But $O(n_k\epsilon(n)/(1+\epsilon(n))) =
O(n\,p(n)\epsilon(n)(1+\epsilon(n))^{k-2}) =
n\,O(p(n)\epsilon(n)(1+\epsilon(n))^{k-2})$, and therefore the kernel set has
size at least $n \,(1-O(p(n) \epsilon(n) (1+\epsilon(n)^{k-2}))$.
By Theorem~\ref{thm:all_must_visit_ranking_graph}, any execution of any deterministic
or randomized Las Vegas local ranking algorithm under any graph exploration model needs $n \,(1-O(p(n) \epsilon(n)
(1+\epsilon(n)^{k-2})))$ queries to solve the instance
$(G_m,\{v_1,\ldots,v_k\})$.
\end{proof}

\subsection{Proof of Theorem~\ref{thm:rankgraph_nph_clique}}
\begin{proof}
CLIQUE asks if a directed graph $G_0$ contains a clique of size $m$ (the
undirected version is straightforwardly reduced to this one). Given a generic
CLIQUE instance $(G_0,m)$, we build an instance for RANKGRAPH that admits a
ranking subgraph of size $2(m+1)+q$, with $q$ depending on $\epsilon$, if and
only if $G_0$ contains a clique of size $m$ (see Figure~\ref{fig:reduction}).
\begin{figure}[h]
\centering
\begin{tikzpicture}[scale=1,
every circle node/.style={inner sep=0pt, minimum size=18},
line/.style={draw, thin},
>=latex,
baseline={(u0.base)}]
  \tikzstyle{every node} = [circle, fill=white];

  \def\Gcolor{black};
  \tikzstyle{G0} = [fill=white, draw=\Gcolor, thick];

  \def\nn{10}
  \def\n{5};
  \def\nsinks{6};

  \node[draw] (u) at (-2,-4) {$u$};
  \path (u) edge[loop below] (u);

  \pgfmathparse{\n - 1};
  \let\m\pgfmathresult;
  \foreach \x in {0, ..., \m} {
    \pgfmathparse{90 + (\x * 360) / \n};
    \let\angle\pgfmathresult;
        \pgfmathparse{int(\x + 1)};
        \let\xx\pgfmathresult
    \node[G0] (u\x) at (\angle : 1) {\color{\Gcolor}$u_{\xx}$};
  }

  \foreach \x in {0, ..., \m} {
    \foreach \y in {1, ..., \m} {
      \pgfmathparse{int(mod(\x + \y, \n))};
      \let\dest\pgfmathresult;
      \draw[G0,->] (u\x) edge[\Gcolor] (u\dest);
    }
  }

  \path [G0,->,\Gcolor, in=215, out=185] (u2) edge[loop] (u2);

  \node[G0] (u5) at (-2.2,-0.4) {\color{\Gcolor}$u_6$};
  \node[G0] (u6) at (-2.2,1) {\color{\Gcolor}$u_7$};
  \node[draw=\Gcolor] (u9) at (2,-0.8) {\color{\Gcolor}$u_8$};

  \path [G0,->,\Gcolor] (u1) edge (u5);
  \path [G0,->,\Gcolor] (u1) edge (u6);
  \path [G0,->,\Gcolor] (u2) edge[bend right] (u5);
  \path (u4) edge[\Gcolor,->] (u9);

  \path [G0,->,\Gcolor] (u5) edge (u2);
  \path [G0,->,\Gcolor] (u6) edge (u2);
  \path [G0,->,\Gcolor] (u6) edge (u5);
  \path [G0,->,\Gcolor] (u6) edge (u0);

  \foreach \x in {1, 2} {
    \node[draw] (s\x) at (-1 + \x, -4) {$s_\x$};
    \path (s\x) edge [loop below] (s\x);
  }
  \foreach \x in {\nsinks} {
    \node[draw] (s\x) at (1 + 0.33 * \x, -4) {$s_\x$};
    \path (s\x) edge [loop below] (s\x);
  }
  \filldraw[black] ($ (s2) !.3! (s\nsinks) $) circle(0.5pt);
  \filldraw[black] ($ (s2) !.50! (s\nsinks) $) circle(0.5pt);
  \filldraw[black] ($ (s2) !.7! (s\nsinks) $) circle(0.5pt);

  \foreach \x in {2, 3, 9} {
    \draw[->,] (u\x) edge (u);
  }

  \draw[->] (u9) edge (s6);
  \draw[->] (u2) edge (s1);

  \draw[->] (u3) edge[loop right] (u3);
  \draw[->] (u3) edge (s1);
  \draw[->] (u3) edge (s2);

  \draw[->] (u9) edge[loop right] (u9);
  \draw[->] (u9) edge (s1);
  \draw[->] (u9) edge (s2);
  \draw[->] (u9) edge (s6);

\end{tikzpicture}
\hspace{1.5cm}
\begin{tikzpicture}[scale=1, every circle
node/.style={draw, thin, inner sep=0pt, minimum size=18},
line/.style={draw, very thin},
>=latex,
baseline={(u0.base)}]
  \tikzstyle{every node} = [circle, fill=white];
  
  \def\nn{10}
  \def\n{5};
  \def\nsinks{2};

  \node[draw] (u) at (-1,-4) {$v$};
  \path (u) edge[loop below] (u);

  \pgfmathparse{\n - 1};
  \let\m\pgfmathresult;
  \foreach \x in {0, ..., \m} {
    \pgfmathparse{90 + (\x * 360) / \n};
    \let\angle\pgfmathresult;
        \pgfmathparse{int(\x + 1)};
        \let\xx\pgfmathresult
        \node[draw] (u\x) at (\angle : 1) {$v_\xx$};
        \path (u\x) edge [in=\angle+15, out=\angle-15, loop] (u\x);
  }

  \begin{scope}[on background layer]
  \foreach \x in {0, ..., \m} {
    \foreach \y in {1, ..., \m} {
      \pgfmathparse{int(mod(\x + \y, \n))};
      \let\dest\pgfmathresult;
      \draw[->] (u\x) edge (u\dest);
    }
  }
  \end{scope}

  \foreach \x in {1, 2} {
    \node[draw] (s\x) at (-0.5 + \x, -4) {$t_\x$};
    \path (s\x) edge [loop below] (s\x);
  }

  \foreach \x in {2, 3} {
    \draw[->,] (u\x) edge (u);
        \draw[->,] (u\x) edge (s1);
        \draw[->,] (u\x) edge (s2);
  }
\end{tikzpicture}
\caption{The graph $G$ built from a CLIQUE instance $(G_0,5)$. The graph $G_0$
contains the nodes and arcs in bold; all the others are added by the reduction
(for simplicity, of the arcs added to $G_0$ we show only those added to $u_3,
u_4$ and $u_8$, and of the arcs going to $v, t_1$ and $t_2$ we show only those
from $v_3$ and $v_4$). Here $u_2$ has outdegree $6$ in $G_0$, and thus $d = 8$.
Since $u \succ v$, there exists a ranking subgraph of size $12$ for $u, v$ in
$G$ if and only if $G_0$ contains a clique of size $5$ (for simplicity, we
assumed $\epsilon$ sufficiently small that $q=0$, i.e.\ $v$ does not have
additional parents besides the clique nodes).}
\label{fig:reduction}
\end{figure}
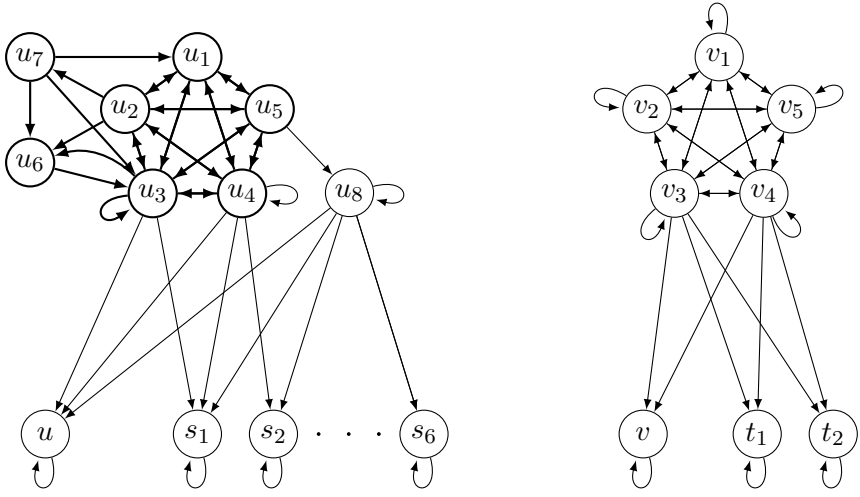

We build the ancestor graph of $u$ starting from $G_0$. First add a self-loop to
$u$. Then, to each node of $G_0$ we add an arc towards $u$, and a self-loop if
not already present (this ``normalizes'' the topology with respect to already
present self-loops). Let $d$ be the maximum outdegree of a node after these
modifications; then bring to $d$ the outdegree of each node by adding a proper
number of arcs towards a set of (at most $d-2$) additional new nodes that act as
sinks. The ancestor graph of $v$ contains $v$ with a self-loop, and a clique of
$m$ nodes each of which has a self-loop, an arc to $v$, and $d - m - 1$ arcs
towards a set of $d - m - 1$ additional nodes that act as sinks (if $d < m+1$
then $G_0$ has maximum outdegree $< m$, in which case we immediately return NO
to CLIQUE without solving RANKGRAPH). In addition, $v$ has a set of $q$ parents
each having outdegree $\lceil d^2/\alpha(1+\epsilon)\rceil$, pointing to $v$ and
to $\lceil d^2/\alpha(1+\epsilon)\rceil - 1$ sink nodes. We choose $q$ such
that, for $|H| = 2(m+1)+q$, the kernel score component $P_H^Q(v)$ provided by
the $q$ parents satisfies $\epsilon P_H^C(v) - \alpha^2(1+\epsilon)/|H|d^2 \leq
P_H^Q(v) \leq \epsilon P_H^C(v)$ where $P_H^C(v)$ is the component provided by
$v$ and its $m$ clique ancestors. A $q$ satisfying this constraint always exists
(and it equals $0$ for $\epsilon$ sufficiently small), since a single parent
gives contribution at most $\alpha^2(1+\epsilon)/|H| d^2$, ensuring sufficient
``resolution''. Since clearly $P_H(v) = P_H^C(v) + P_H^Q(v)$, we have
$(1+\epsilon)(P_H^C(v)-\alpha^2/|H|d^2) \leq P_H(v) \leq(1+\epsilon)P_H^C(v)$.

Denote by $G$ the resulting graph, and consider the input instance
$(G,u,v,2(m+1)+q)$ for RANKGRAPH. The intuition is that a ranking subgraph of
kernel size $2(m+1)+q$ exists if and only if there are $m$ ancestors of $u$ that
induce a subgraph ``strong enough'' to guarantee $P_H(u) \geq
P_H(v)/(1+\epsilon)$ -- and this happens exactly if they form a clique. Indeed
we prove:
\begin{enumerate}
 \item If $G_0$ contains a clique of size $m$, then $G$ has a compatible ranking
subgraph of size $2(m+1)+q$ for $u \succ v$. Indeed, consider the visit subgraph
having in the kernel set the $m$ clique nodes, along with $u$, $v$, and all the
$m + q$ ancestors of $v$. This visit subgraph has kernel size $2(m+1)+q$ and
also satisfies property 2 of Theorem~\ref{thm:rankgraph_char}. Observe that
$P_H(u) = P_H^C(v)$ since the subgraphs induced by $u$ with its $m$ ancestors
and by $v$ with its $m$ clique ancestors are identical -- cliques of $m$ nodes
with outdegree $d$, each having in addition a self-loop and an arc to the
target node. But from above we have $P_H^C(v) \geq P_H(v)/(1+\epsilon)$, and
therefore
$P_H(u) \geq P_H(v)/(1+\epsilon)$; thus $H$ satisfies also property 1 of
Theorem~\ref{thm:rankgraph_char} and is a ranking subgraph of kernel size
$2(m+1)+q$ for $u \succ v$ compatible with $G$.
\item If $G_0$ does not contain a clique of size $m$, then either $P(u) <
P(v)/(1+\epsilon)$ in $G$ (which we can check by computing their scores, and in
which case we immediately return NO to CLIQUE without solving RANKGRAPH) or
there is no ranking subgraph of kernel size $2(m+1)+q$ for $u \succ v$ in $G$,
as we prove. A ranking subgraph $H$ for $u \succ v$ must include $u$ and $v$ by
definition, and the $m + q$ ancestors of $v$ by property 2 of
Theorem~\ref{thm:rankgraph_char}; thus, to have kernel size $2(m+1)+q$, it must
contain at most $m$ ancestors of $u$. If these $m$ ancestors do not induce a
clique, at least one arc is missing between two of them (it instead points to a
frontier node, and any path using it gives null kernel contribution). If these
nodes formed a clique, this arc would provide a kernel contribution via paths
starting with it and leading to $u$ via another arc (followed by zero or
more self-loops) amounting to $\alpha^2/|H|d^2$; and additional contribution
via paths that use it as intermediate arc. The kernel score of $u$ in this
``clique'' scenario is $P_H^C(u) = P_H^C(v)$, and thus we have $P_H(u) <
P_H^C(v) - \alpha^2/|H|d^2$. With the inequalities above this implies
$P_H(u)<P(v)/(1+\epsilon)$, inducing $v \succ u$ and proving that $H$ is not a
ranking subgraph for $u \succ v$.
\end{enumerate}
Therefore we have a (clearly polynomial) reduction from CLIQUE to RANKGRAPH,
which concludes the proof.
\end{proof}

\subsection{Proof of Theorem~\ref{thm:rankgraph_nph_dominating}}
\begin{proof}
Given a graph $G_0$ and an integer $m > 0$, DOMINATING SET asks if there is a
subset $D$ of $m$ nodes in $G_0$ such that any other node in $G_0$ points to a
node in $D$ (the undirected version can be straightforwardly reduced to this
one).

Starting from a generic directed DOMINATING SET instance $(G_0,m)$, we
build a graph $G$ (Figure~\ref{fig:reduction_ds}) such that a compatible ranking
subgraph of size $3 + m$ exists for $u \succ v$ if and only if $G_0$ contains a
dominating set of size $m$. We assume that $G_0$ has no dangling nodes;
otherwise, the dominating set is formed by all the dangling nodes plus the
dominating set on the remaining nodes.
\begin{figure}[h]
\centering
\begin{tikzpicture}[scale=1, every circle
node/.style={draw, thin, inner sep=0pt, minimum size=18},
line/.style={draw, thin},
>=latex,
baseline={(u.base)}]
  \tikzstyle{every node} = [circle, fill=white];

  \node[draw] (u) at (0,0) {$u$};
  \path (u) edge[loop below] (u);

  \node[draw] (v) at (4,0) {$v$};
  \path (v) edge[loop below] (v);

  \node[draw] (w) at (4,2) {$w$};
  \path (w) edge[->] (v);

  \def\spacing{1};

  \def\ycoord{2};
  \def\nanc{5};
  \pgfmathparse{\nanc/2 + 1/2};
  \let\hoff\pgfmathresult
  \foreach \i in {1,...,\nanc} {
	\pgfmathparse{(\i - \hoff) * \spacing};
	\let\xcoord\pgfmathresult;
	\node[draw] (u\i) at (\xcoord, \ycoord) {$u_{\i}$};
	\path[->] (u\i) edge (u);
  }

  \def\ycoord{4};
  \def\nanc{5};
  \pgfmathparse{\nanc/2 + 1/2};
  \let\hoff\pgfmathresult
  \foreach \i in {1,...,\nanc} {
	\pgfmathparse{(\i - \hoff) * \spacing};
	\let\xcoord\pgfmathresult;
	\node[draw] (u1\i) at (\xcoord, \ycoord) {$u'_{\i}$};
	\path[->] (u1\i) edge[bend left=40] (w);
  }

  \path (u11) edge[->] (u3);
  \path (u12) edge[->] (u1);
  \path (u13) edge[->] (u1);
  \path (u13) edge[->] (u4);
  \path (u13) edge[->] (u5);
  \path (u14) edge[->] (u3);
  \path (u15) edge[->] (u3);
\end{tikzpicture}
\caption{The graph $G$ built from a DOMINATING SET instance $(G_0,2)$. The graph
$G_0$ contains $5$ nodes $v_1,\ldots,v_5$ and the arc $(u'_i,u_j)$ exists in
$G$ if and only if $(v_i,v_j)$ exists in $G_0$. Graph
$G$ admits a compatible ranking subgraph with kernel size $5$ if and only if
$G_0$ admits a dominating set of size $2$. In this example, such a ranking
subgraph has $u, u_1, u_3, v, w$ as its kernel set, and thus a dominating
set of size $2$ for $G_0$ is $\{v_1, v_3\}$.}
\label{fig:reduction_ds}
\end{figure}
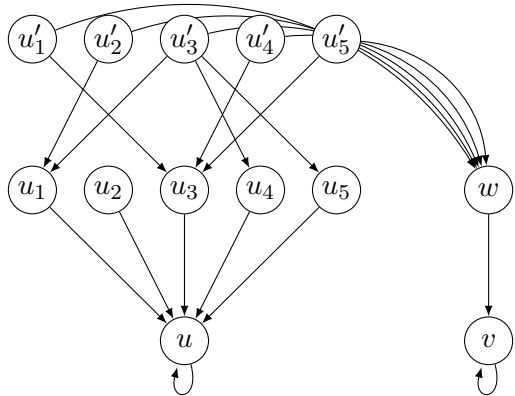
Denote the nodes of $G_0$ by $v_1, \ldots, v_n$. The graph $G$ contains $u$ and
$v$ with their self-loops; $n$ parents $u_1, \ldots, u_n$ of $u$; one parent $w$
of $v$; and $n$ nodes $u'_1,\ldots,u'_n$ where $u'_i$ has one arc towards $w$
and, for every $j = 1,\ldots,n$, an arc towards $u_j$ if and only if the arc
$(v_i,v_j)$ exists in $G_0$. The intuition is that a ranking subgraph of
kernel size $3 + m$ exists if and only if there are $m$ parents of $u$ that, if
in the kernel set, can ``cover'' the frontier nodes among $u_1',\ldots,u_n'$
ensuring that none of them gives more contribution to $v$ than to $u$. Indeed we
prove:
\begin{enumerate}
 \item If $G_0$ contains a dominating set of size $m$, then there exists a
ranking subgraph of size $3 + m$ for $u \succ v$ in $G$. Let $D$ be the
dominating set of size $m$ in $G_0$. Consider the visit subgraph $H$ having in
its kernel set $u$, $v$, $w$ and, for $i = 1,\ldots,n$, node $u_i$ if and only
if $v_i$ is in the dominating set, thus having all the other nodes of $G$ as
frontier nodes.
Since $m > 0$, then $u$ has at least one parent in the kernel set and $P_H(u)
\geq P_H(v)$; therefore $H$ satisfies property 1 of
Theorem~\ref{thm:rankgraph_char} for any $\epsilon \geq 0$. Since each node
$u'_i$ points to at least a parent of $u$ and to exactly one parent of $v$, then
$H$ also satisfies property 2 of Theorem~\ref{thm:rankgraph_char}, proving that
it is a ranking subgraph of kernel size $3+m$ for $u \succ v$ compatible with
$G$.
 \item If there exists a ranking subgraph of size $3+m$ for $u \succ v$ in
$G$, then we can build a dominating set of size $m$ for $G_0$ as follows. The
kernel set of any ranking subgraph for $u \succ v$ in $G$ contains $v$ and $w$
by property 2 of Theorem~\ref{thm:rankgraph_char}, and contains $u$
by definition. Therefore, if the kernel set has size $3 + m$, it must include
exactly $m$ nodes among $u_1,\ldots,u_n,u'_1,\ldots,u'_n$. Again by property 2
of Theorem~\ref{thm:rankgraph_char}, either the generic node $u'_i$ is in the
kernel set, or one of its children  $u_j$ is in the kernel set, otherwise
$P_H(u'_i,u) = 0 < P_H(u'_i,v)$. For each $u'_i$ in the kernel set, we
replace it with one of its children (even if already in the kernel set); this
does not increase the kernel size and maintains property 2 of
Theorem~\ref{thm:rankgraph_char}. We end up with a visit subgraph whose kernel
set contains $m' \leq m$ nodes among $u_1,\ldots,u_n$ that ``cover'' (i.e.\
have an arc from) each of the frontier nodes $u'_1,\ldots,u'_n$. A dominating
set of size $m$ for $G_0$ contains $v_i$ if and only if $u_i$ is in the kernel
set, and $m-m'$ other nodes in $G_0$.
\end{enumerate}
Therefore we have a (clearly polynomial) reduction from DOMINATING SET to
RANKGRAPH, which concludes the proof.
\end{proof}

\subsection{Proof of Lemma~\ref{lem:inapproximability}}
\begin{proof}
Given an instance $(G_0,m)$ for DOMINATING SET, the reduction builds an instance
$(G,u,v,3+m)$ for RANKGRAPH where $|G| = 2|G_0| + 3$. Furthermore, a solution
of size $3+m$ for RANKGRAPH exists if and only if there is a solution of size
$m$ for DOMINATING SET. This implies that approximating RANKGRAPH within a small
(surely up to polynomial in $|G|$) factor would allow approximating DOMINATING
SET within the same factor factor. Approximating DOMINATING SET within
$o(\log(|G_0|))$ is NP-hard~\cite{Raz&1997}, thus approximating RANKGRAPH within
$o(\log(|G|))$ is also NP-hard.
\end{proof}

\subsection{Proof of Theorem~\ref{thm:avoid_ranking_graphs}}
\label{apx:thm_avoid_ranking_graphs}
\begin{proof}
Let $m = \lceil 1/\eta \rceil$. The graph $G$ is formed by $u$ and $v$, $m+1$ orphan nodes that point only to $u$, and
$m$ orphan nodes that point only to  $v$. Therefore, $P(u) > P(v)$. The algorithm, EludeRS, is given in the next page.
\begin{algorithm}[h]
\caption{$EludeRS(G,\{u,v\})$}
\label{alg:magic}
\begin{algorithmic}
\State Call \links on $u$ and $v$
\If {($u$ has $m+1$ parents and $0$ children) and ($v$ has $m$ parents and $0$
children)}
   \State Call \links on all the parents of $u$
   \State Call \links on all the parents of $v$ except one chosen uniformly at
random
   \If {(all the queried parents are orphans and have one child)}
     \State \Return $u \succ v$
   \Else
     \State Visit the entire ancestor graph of $u$ and $v$ and return the
correct ranking
   \EndIf
\Else
   \State Visit the entire ancestor graph of $u$ and $v$ and return the
correct ranking
\EndIf
\end{algorithmic}
\end{algorithm}

It is immediate to verify that EludeRS on $(G,\{u,v\})$ correctly returns $u
\succ v$ without querying all the parents of $v$.
By Theorem~\ref{thm:rankgraph_char}, all the parents of $v$ must belong to the
kernel set of any ranking subgraph for $u,v$ in $G$. Therefore, EludeRS never
visits such a ranking subgraph.
It is also immediate to verify that, on any instance containing a graph
different  from $G$, EludeRS visits all the ancestors of $u$ and $v$ and returns
the correct
ranking; unless $u$ and $v$ have zero children and respectively $m+1$ and $m$
parents, all these parents are orphans and have one child with the exception of
one of $v$'s parents, and this parent is exactly that not queried by EludeRS.
This happens with probability at most $1/m \leq \eta$, and thus EludeRS has
confidence $1 - \eta$.

Note that the score separation is $\epsilon = \alpha/(1+\alpha m)$ in the simple
instance above, but any value of $\epsilon$ can be approximated by adding parents
to $u$ or (with additional children) to $v$, and tuning EludeRS accordingly.
\end{proof}

\subsection{Proof of Theorem~\ref{thm:cost_monte_carlo_onlycrawl}}
\label{apx:proof_thm_cost_monte_carlo_onlycrawl}
\begin{proof}
We build a graph $G_m$ of size $n = \Theta(m)$ satisfying the statement for
every $m$ sufficiently large. $G_m$ consists of $k$ disjoint subgraphs $G_m^1,
\ldots, G_m^k$. Subgraph $G_m^i$ (Figure~\ref{fig:cost_montecarlo_onlycrawl})
contains the target node $v_i$ and its self-loop, a chain of $\log_\alpha p(m)$
nodes ending with $v_i$, and $m$ nodes having the head of the chain as their
sole child; one of them, called ``strong ancestor'', has $((1+\epsilon)^{i-1} -
1) m / \alpha$ exclusive parents. All these quantities are meant as rounded to
the nearest integer -- we assume $m$ to be sufficiently large to avoid dealing
with rounding operators and still guarantee the asymptotic results of the
statement. Intuitively, MC can only explore the graph ``from below'' to find at
least one strong ancestor, which happens with probability proportional to the
fraction of nodes visited, otherwise the returned ranking is correct with
probability at most $1/k!$.

\begin{figure}[h]
 \centering
 \begin{tikzpicture}[scale=1, every circle
node/.style={draw, thin, inner sep=0pt},
line/.style={draw, very thin},
>=latex,
baseline={(vi.base)}]
\tikzstyle{ancestor} = [circle, draw, fill=white, minimum size=18];
\tikzstyle{strong} = [circle, draw, fill=lightgray, minimum size=18];
\tikzstyle{target} = [circle, draw, fill=white, minimum size=18];

\def\spacing{0.8};
\def\ybasecoord{0};
\def\xbasecoord{5};
\def\ancybasecoord{3};

\node[target] (vi) at (\xbasecoord, \ybasecoord) {$v_i$};
\path[->] (vi) edge[loop below] (vi);

\path (vi) ++(0, 1) node[ancestor] (x1) {};
\path[->] (x1) edge (vi);
\path (x1) ++(0, 1.5) node[ancestor] (x2) {};
\path[->] (x2) edge[dotted] (x1);
\path (x2) ++(0, 1) node[ancestor] (x3) {};
\path[->] (x3) edge (x2);
\path (x1) ++(1.5,0.75) node {$\log_\alpha p(m)$};

%
\def\nanc{7};
\pgfmathparse{(\nanc + 1)/2};
\let\hoff\pgfmathresult
\foreach \j in {1,...,\nanc} {
  \pgfmathparse{(\j - \hoff) * \spacing};
  \let\xcoord\pgfmathresult;
  \path (x3) ++(\xcoord, 1.5) node[ancestor] (w\j) {};
  \path[->] (w\j) edge (x3);
}
\path (w\nanc) ++(1,0) node {$m$};
\def\strongparent{5};
\path (w\strongparent) ++(0,0) node[strong] {};

\def\ncloud{4};
\pgfmathparse{(\ncloud + 1)/2};
\let\hoff\pgfmathresult
\foreach \j in {1,...,\ncloud} {
  \pgfmathparse{(\j - \hoff) * \spacing};
  \let\xcoord\pgfmathresult;
  \path (w\strongparent) ++(\xcoord, 1.5) node[ancestor] (z\j) {};
  \path[->] (z\j) edge (w\strongparent);
}
\path (z\ncloud) ++(1.8,0) node {$\frac{(1 + \epsilon)^{i-1} - 1}{\alpha} m$};

\end{tikzpicture}
 \caption{Subgraph $G_m^i$ of graph $G_m$
(Theorem~\ref{thm:cost_monte_carlo_onlycrawl}). The bulk of $v_i$'s PageRank
score is provided by its $m$ ancestors via a chain of $\log_\alpha
p(m)$ nodes (damping their contribution by a factor $p(m)$), while its
relative ranking is precisely determined by the number of parents of its
``strong ancestor'' (filled in grey). Any (Monte Carlo) algorithm without
global queries can at best crawl backward from $v_i$ and sample the $m$
ancestors among which lies the strong ancestor.}
\label{fig:cost_montecarlo_onlycrawl}
\end{figure}
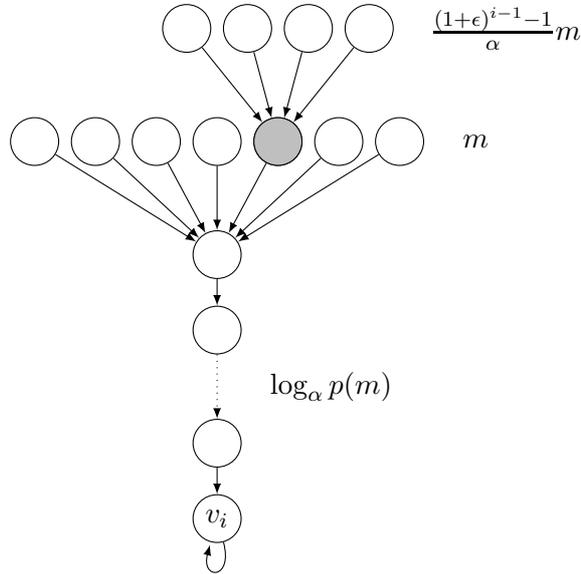

$G_m$ has size:
\begin{align}
n &= \sum_{i=1}^k \bigg( \log_\alpha p(m) + m + m \frac{(1 + \epsilon)^{i-1} -
1}{\alpha} \bigg)
\end{align}
which is in $\Theta(m)$ since $\log_\alpha p(m) \in O(\log_\alpha(1/m))
= O(m)$ and $\alpha, \epsilon, k \in \Theta(1)$. The PageRank score of $v_i$ is
\begin{align}
P(v_i) &= \frac{1}{n}\bigg(\sum_{j=0}^{\log_\alpha p(m)-1} \!\!\!\!\!\! \alpha^j +
\alpha^{\log_\alpha p(m)} \Big(m + \alpha \frac{(1 + \epsilon)^{i-1} - 1}{\alpha}
m\Big)\bigg) \approx p(m) (1 + \epsilon)^{i-1}
\end{align}
which is in $\Theta(p(n))$ since $\alpha, \epsilon, k \in \Theta(1)$ and
$n=\Theta(m)$. The scores satisfy $P(v_{i+1}) \approx (1 + \epsilon)P(v_i)$,
making the target nodes $\approx \epsilon$-separated (a strict
$\epsilon$-separation is obtained giving $((1+c\epsilon)^{i-1} - 1) m/\alpha$
parents to $v_i$'s strong ancestor for some constant $c$, but we avoid it to
keep the proof simple).

We prove the lower bound, showing that it holds even if the algorithm is
aware of the structure of the graph, except for the identities and the number
of parents of the strong ancestors. Since MC cannot perform global
queries (which may lead to the parents of the strong ancestors or the strong
ancestors themselves), it must explore the graph via local queries starting
from the target nodes, ``crawling backward'' along the chains, and directly
querying the $\Theta(m)$ ancestors, essentially sampling them (with e.g.\
\linksNoSp) until finding the strong ancestor. Note that, for a non-uniform
sampling, there exists a worst-case instance (having the least-queried ancestor
as the strong ancestor) requiring a strictly higher number of queries than for a
uniform sampling -- therefore we assume
that MC performs a uniform sampling. Suppose MC performs $q$ queries overall. 
The probability $\gamma$ that MC finds at least one strong ancestor, in which
case it returns the correct ranking with probability at most $1$, satisfies then
$\gamma \leq q/m$. With probability $1 - \gamma$, MC does not visit any strong
ancestor and must return an arbitrary ranking -- the strategy minimizing the
worst-case error probability consists again in returning a ranking uniformly at
random among the $1/k!$ possible rankings. The overall probability of returning
the correct ranking is thus upper bounded by
\begin{align}
 \gamma \, 1 + (1 - \gamma)\frac{1}{k!} \leq \frac{1}{k!} + \gamma
\end{align}
Since MC has confidence $1/k! + \delta$, we have
$\delta \leq \gamma$, and since $\gamma \leq q/m$ we have $q \geq
\delta m \in \Omega(\delta n)$.
\end{proof}

\subsection{Proof of Lemma~\ref{lem:sample}}
\begin{proof}
Note that SampleNode performs a random walk on the graph $G$ as if it was preprocessed by adding, to each dangling node,
$n$ outgoing links to itself and to every node (this is emulated by the call to \jump made when \crawl returns an empty
subgraph).
The probability that SampleNode follows exactly $\tau$ arcs is $(1-\alpha)\alpha^\tau$. Conditioned on this event,
the probability that SampleNode returns $v$ equals $\frac{1}{n}\sum_{z\in G} \operatorname{inf}_{\tau}(z,v)$ (recall
from Section~\ref{sec:pagerank} that $\operatorname{inf}_{\tau}(z,v)$ denotes the probability that, starting from node
$z$ and following at each step a random outgoing link of the current node, after exactly $\tau$ steps one is at $v$).
Then the probability that SampleNode returns $v$ equals $\sum_{\tau=0}^\infty (1-\alpha)\alpha^\tau 
\frac{1}{n}\sum_{z \in G}\operatorname{inf}_{\tau}(z,v)$, which coincides with $P(v)$ by
Equation~(\ref{eqn:pagerank_series_influences}).
\end{proof}

\subsection{Proof of Lemma~\ref{lem:sampletime}}
\begin{proof}
Each time SampleNode goes through the first instruction of the loop, it terminates immediately with an independent
probability $(1-\alpha)$; and, by the $j$-th time, it has issued at most $1 + 2(j - 1) \leq 2j$ queries.
Thus, SampleNode terminates after issuing at most $\frac{2}{1-\alpha}$ queries in expectation.
If SampleNode issues more than $\frac{2(1+\Delta)m}{1-\alpha}$ queries, then it goes through the first instruction in
the loop more than $\frac{(1+\Delta)m}{1-\alpha}$ times; the probability of this happening over $m$ calls is the
probability that that instruction makes SampleNode terminate less than $m = (1 - \frac{\Delta}{1 + \Delta}) \cdot (1 +
\Delta)m$ out of $\frac{(1+\Delta)m}{1-\alpha}$ times.
By a simple Chernoff bound, this probability is less than $e^{-\frac{(1+\Delta)m(\frac{\Delta}{1+\Delta})^2}{2}}
=e^{-\frac{m}{2}\cdot\frac{\Delta^2}{1+\Delta}}$.
\end{proof}

\subsection{Proof of Theorem~\ref{thm:samplerank2}}
\begin{proof}
By construction, the ranking returned by SampleRank is incorrect only if some
scores fall outside of their $1-\eta/k$ confidence intervals, which by a union
bound happens with probability at most $k \cdot \eta/k = \eta$. Therefore, the
ranking is correct with probability at least $1 - \eta$. The rest of the proof
shows that SampleRank performs more than $O\big(\frac{1}{1 -
\alpha}\log(\frac{4k}{\eta}) \frac{(1+\epsilon)^2}{\epsilon^2 p}\big)$ queries
with probability at most $\eta$. Precisely, we show that with probability at
least $1-\eta/2$ it performs at most
$O\big(\log(\frac{4k}{\eta})\frac{(1+\epsilon)^2}{\epsilon^2 p}\big)$ calls
to SampleNode, and that with probability at least $1-\eta/2$ these require
$O(\frac{1}{1-\alpha})$ as many queries -- yielding the desired bound
with probability at least $(1 - \eta/2)^2 \geq 1 - \eta$.

Consider two target nodes $u, v$ and suppose that $P(v)=(1+\epsilon/2)P(u)$.
The intuition is that, after a sufficient number of samples, with probability
at least $1-\eta/2k$ the confidence intervals obtained from the estimates
$\hat P_m(u)$ and $\hat P_m(v)$ either are disjoint (which holds also for all
$P(v) \geq (1+\epsilon/2)P(u)$) or one is within $1+\epsilon$ of the other
(which holds also for all $P(v) \leq (1+\epsilon/2)P(u)$).
We first bound the probability that the intervals become disjoint.
The probability that, after $m$ samples, a node with score at least $P(u)(1 +
\epsilon/4)$ receives a score estimate smaller or equal to
$P(u)(1 + \epsilon/8) = P(u)(1 + \epsilon/4)(1- \frac{\epsilon/8}{1 +
\epsilon/4})$ is, by a Chernoff bound, less than
$e^{-mP(u)(1+\epsilon/4)\frac{(\epsilon/8)^2}{2(1+\epsilon/4)^2}} =
e^{-\Theta(mP(u)\frac{\epsilon^2}{1 + \epsilon})}$.
For a proper $m = O(\log(\frac{2k}{\eta}) \frac{1+\epsilon}{\epsilon^2
P(u)})$, this probability is upper bounded by $\eta/2k$, and thus if $\hat
P_m(u) \leq P(u)(1 + \epsilon/8)$, the confidence interval for $P(v)$ cannot
contain the value $P(u)(1+\epsilon/4)$.
By a Chernoff bound, $\hat P_m(u) > P(u)(1 + \epsilon/8)$ with probability less
than $e^{-mP(u)\frac{(\epsilon/8)^2}{2 + \epsilon/8}} =
e^{-\Theta(mP(u)\frac{\epsilon^2}{1 + \epsilon})}$, which again for a proper $m
= O(\log(\frac{2k}{\eta}) \frac{1+\epsilon}{\epsilon^2 P(u)})$ is upper
bounded by $\eta/2k$. Therefore, after $O(\log(\frac{2k}{\eta})
\frac{1+\epsilon}{\epsilon^2 P(u)})$ samples, with probability at
least $1 - \eta/2k$ the confidence interval for $P(u)$ does not contain
$P(u)(1+\epsilon/4)$. A similar calculation shows that, after
$O(\log(\frac{2k}{\eta}) \frac{1+\epsilon}{\epsilon^2 P(u)})$ samples, with
probability $1 - \eta/2k$ also the confidence interval for $P(v)$ does not
contain $P(u)(1+\epsilon/4)$, and thus $C_m(u) \cap C_m(v) = \emptyset$.
A similar argument shows that, after $O(\log(\frac{2k}{\eta})
\frac{(1+\epsilon)^2}{\epsilon^2 P(u)})$ samples, with probability at least
$1 - \eta/2k$ the confidence interval of $P(u)$ does not include values below
$P(u)(1 - \epsilon/4(1+\epsilon))$ and that of $P(v)$ does not include values
above $P(v)(1 + \epsilon/(8 + 4\epsilon))$. The ratio between these two extremes
is exactly $1+\epsilon$, and therefore with probability at least $1-\eta/2k$ we
have $P^U_m(v)/ P^L_m(u) \leq 1 + \epsilon$.
Hence, after $O(\log(\frac{2k}{\eta})\frac{(1+\epsilon)^2}{\epsilon^2
P(u)})$ samples we have that, ideally ordering the nodes by their true scores,
each pair of consecutive nodes has intervals that overlap or are not within
$1+\epsilon$ one of the other with probability at most $\eta/2k$. By a union
bound, then, with probability $1 - \eta/2$ SampleRank terminates after
$O(\log(\frac{2k}{\eta})\frac{(1+\epsilon)^2}{\epsilon^2 P(u)})$ iterations
(and the probability that it exceeds the expected number of iterations by a
factor $c \geq 1$ decreases as $e^{-c}$).

By Lemma~\ref{lem:sampletime}, the probability that
$O(\log(\frac{2k}{\eta})\frac{(1+\epsilon)^2}{\epsilon^2 P(u)})$
calls to SampleNode perform more than $\Theta(\frac{1+\Delta}{1-\alpha}
\log(\frac{4k}{\eta}) \frac{(1+\epsilon)^2}{\epsilon^2 p})$
queries is bounded by
$e^{-\Theta(\log(\frac{4k}{\eta}) \frac{(1+\epsilon)^2}{\epsilon^2
p}\frac{\Delta^2}{1 + \Delta})}
= O((\eta/4k)^{\frac{(1+\epsilon)^2 \Delta^2}{\epsilon^2 p
(1+\Delta)}})$, which is strictly less than $\eta/2$ for appropriate constants
(and e.g.\ $\Delta = 2$). Therefore SampleRank performs $O(\frac{1}{1-\alpha}
\log(\frac{4k}{\eta}) \frac{(1+\epsilon)^2}{\epsilon^2 p})$
with probability $(1-\eta/2)^2 > 1-\eta$. This concludes the proof.
\end{proof}

\subsection{Proof of Theorem~\ref{thm:cost_monte_carlo}}
\begin{proof}
We build a graph $G_m$ of size $n = \Theta(m)$ satisfying the statement for
every $m$ sufficiently large, first for the case $p(x) \geq \Theta(x^{-2/3})$,
then extending the result to $p(x) \geq \Theta(1/x)$. $G_m$ consists of $k$
disjoint subgraphs $G_m^1, \ldots, G_m^k$. Subgraph
$G_m^i$ (Figure~\ref{fig:cost_montecarlo}) contains the target node $v_i$ with
a self-loop, and $m \sqrt{p(m)}/\alpha$ parents of $v_i$; one of these,
called ``strong ancestor'', has $v_i$ as its sole child and has $m \,
p(m)((1+\epsilon)^{i-1}-1)/\alpha^2$ orphan parents, and all the others have
outdegree $1/\sqrt{p(m)}$, pointing to $v_i$ and to $1/\sqrt{p(m)} - 1$
exclusive sink nodes (to keep Figure~\ref{fig:cost_montecarlo} legible, we have
explicitly drawn these sinks only for one parent). All these quantities are
meant as rounded to the nearest integer -- we assume $m$ to be sufficiently
large to avoid dealing with rounding operators and still guarantee the
asymptotic results of the statement.

\begin{figure}[h]
 \centering
 \begin{tikzpicture}[scale=1, every circle
node/.style={draw, thin, inner sep=0pt},
line/.style={draw, very thin},
>=latex,
baseline={(vi.base)}]
\tikzstyle{ancestor} = [circle, draw, fill=white, minimum size=18];
\tikzstyle{strong} = [circle, draw, fill=lightgray, minimum size=18];
\tikzstyle{target} = [circle, draw, fill=white, minimum size=18];

\def\spacing{0.8};
\def\ybasecoord{0};
\def\xbasecoord{5};
\def\ancybasecoord{3};

\node[target] (vi) at (\xbasecoord, \ybasecoord) {$v_i$};
\path[->] (vi) edge[loop below] (vi);

%
\def\nanc{7};
\pgfmathparse{(\nanc + 1)/2};
\let\hoff\pgfmathresult
\foreach \j in {1,...,\nanc} {
  \pgfmathparse{(\j - \hoff) * \spacing};
  \let\xcoord\pgfmathresult;
  \path (vi) ++(\xcoord, 2.5) node[ancestor] (w\j) {};
  \path[->] (w\j) edge (vi);
}
\path (w\nanc) ++(1.2,0) node {$\frac{m \sqrt{p(m)}}{\alpha}$};
\def\strongparent{5};
\path (w\strongparent) ++(0,0) node[strong] {};

\def\ncloud{4};
\pgfmathparse{(\ncloud + 1)/2};
\let\hoff\pgfmathresult
\foreach \j in {1,...,\ncloud} {
  \pgfmathparse{(\j - \hoff) * \spacing};
  \let\xcoord\pgfmathresult;
  \path (w\strongparent) ++(\xcoord, 1.5) node[ancestor] (z\j) {};
  \path[->] (z\j) edge (w\strongparent);
}
\path (z\ncloud) ++(2,0) node {$\frac{m p(m)
((1+\epsilon)^{i-1}-1)}{\alpha^2}$};

\def\nsink{3};
\pgfmathparse{(\nsink + 1)/2};
\let\hoff\pgfmathresult
\foreach \j in {1,...,\nsink} {
  \pgfmathparse{(\j - \hoff - 1) * \spacing};
  \let\xcoord\pgfmathresult;
  \path (w1) ++(\xcoord, -1.5) node[ancestor] (s\j) {};
  \path[->] (w1) edge (s\j);
}
\path (s1) ++(-1.5,0) node {$\frac{1}{\sqrt{p(m)}} - 1$};

\end{tikzpicture}
 \caption{Subgraph $G_m^i$ of graph $G_m$ (Theorem~\ref{thm:cost_monte_carlo}).
The bulk of $v_i$'s PageRank score is provided by its $\approx m\sqrt{p(m)}$
parents, while its relative ranking is precisely determined by the number of
parents of its ``strong ancestor'' (filled in grey). Any Monte Carlo algorithm
can find a strong ancestor using either global queries to directly discover it
or one of its parents (which are a fraction $O(p(n))$ of all nodes), or local
queries to sample $v_i$'s parents (which are $\Omega(n^{2/3})$).}
 \label{fig:cost_montecarlo}
\end{figure}
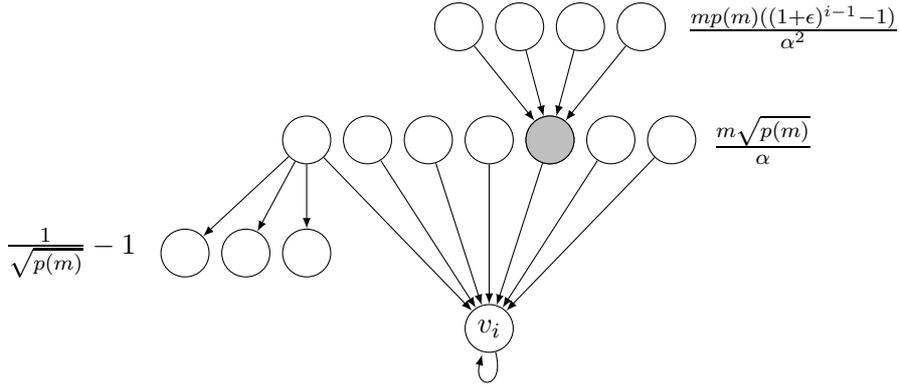

$G_m$ has size:
\begin{align}
n &\approx \sum_{i=1}^k \bigg( \frac{m\sqrt{p(m)}}{\alpha} \frac{1}{\sqrt{p(m)}}
+ \frac{m\,p(m)((1+\epsilon)^{i-1}-1)}{\alpha^2} \bigg)\nonumber\\
 &= m \bigg( \frac{k}{\alpha} + p(m) \sum_{i=1}^k
\frac{((1+\epsilon)^{i-1}-1)}{\alpha^2} \bigg)
\end{align}
which is in $\Theta(m)$ since $\epsilon, \alpha, k \in \Theta(1)$ and $p(m) \in
O(1)$.

The PageRank score of $v_i$ is:
\begin{align}
 P(v_i) &= \frac{1}{n}\bigg(1 + \alpha \Big(\frac{m\sqrt{p(m)}}{\alpha} - 1\Big)
\frac{1}{1/\sqrt{p(m)}} + \alpha + \alpha^2
\frac{m\,p(m)((1+\epsilon)^{i-1}-1)}{\alpha^2}\bigg)\nonumber\\
&\approx p(m)(1+\epsilon)^{i-1}
\end{align}
which is in $\Theta(p(n))$ since $\epsilon \in O(1)$ and $n = \Theta(m)$. The
scores satisfy
$P(v_{i+1}) \approx (1 + \epsilon)P(v_i)$, making the target nodes $\approx
\epsilon$-separated (a strict $\epsilon$-separation is obtained giving
$ m p(m) ((1+c\epsilon)^{i-1} - 1) / \alpha^2$ parents to $v_i$'s strong
ancestor for some constant $c$, but we avoid it to keep the proof simple).

We prove the lower bound, showing that it holds even if the algorithm is aware
of the structure of the graph, except for the identities and the number of
parents of the strong ancestors. To prove the bound under \emph{any}
graph exploration model, we adopt the ``most powerful'' local exploration model
possible, which allows \jump and \links queries as we prove immediately.
To find at least one strong ancestor, MC can:
\begin{enumerate}
 \item Explore the graph using global queries (recall that global queries do not
take any input and return in output a single node depending on the set of nodes
of the graph, but not on its arcs). Since the output of a global query does not
depend on what the algorithm has already visited, we imagine that all the global
queries have been performed before all the local queries, maximizing the
information obtained by MC. Observe now that \jump is the ``less adversarial''
type of global query possible. Indeed, for any global query returning nodes
not uniformly at random, we can reduce the amount of information obtained by MC
by permuting the parents of a target node so that the strong ancestor and its
parents are the nodes returned less likely. Even if MC is aware of this and
tries to use global queries to ``exclude'' parents that are less likely to be
strong ancestors, it needs $\Omega(n)$ queries to restrict to $O(n)$ candidates.
Thus, we assume that \jump is allowed. In this case, the probability that any
\jump query returns a strong ancestor or one of its parents (in which case the
algorithm can identify the strong ancestor with only an additional \links query)
is in $O((n\,p(n)((1+\epsilon)^{i-1}-1)/\alpha^2)/n) = O(p(n))$, and for $q$
queries it
is in $O(q p(n))$.
 \item Explore the graph using local queries. Note that a randomized algorithm
using \links can emulate any other type of local query; thus, we assume that \links
is allowed. MC can then start from the target nodes and sample their parents
-- the only ancestors that it knows besides those returned by \jump --
until finding the strong ancestor. For a non-uniform sampling, there
exists a worst-case instance (having the least-queried ancestor as the strong
ancestor) requiring a strictly higher number of queries than for a uniform
sampling -- therefore we assume that MC performs a
uniform sampling. The probability that querying $q$ parents yields a
strong ancestor is in $O(q/n\sqrt{p(n)})$ which, since $p(n) \in
\Omega(1/n^{2/3})$, is in $O(q n^{-\frac{2}{3}})$.
\end{enumerate}
Therefore, if MC performs $O(q)$ queries, the probability $\gamma$
of finding at least one strong ancestor, in which case it returns the correct
ranking with probability at most $1$, is $O(q \operatorname{max}\{p(n),
n^{-2/3}\})$; and hence $q \in \Omega(\gamma /
\operatorname{max}\{p(n), n^{-2/3}\})$. With probability $1 - \gamma$, MC does
not find any strong ancestor, and then outputs an arbitrary ranking -- again,
the strategy minimizing the worst-case error probability consists in returning a
ranking uniformly at random among the $k!$ possible rankings, which therefore
is correct with probability at most $1/k!$. The overall probability of
returning a correct ranking is thus upper bounded by
\begin{align}
 \gamma \, 1 + (1 - \gamma)\frac{1}{k!} \leq \frac{1}{k!} + \gamma
\end{align}
Since MC has confidence $1/k! + \delta$, we have
$\delta \leq \gamma$, and therefore $q \in \Omega(\delta
/ \operatorname{max}\{p(n), n^{-2/3}\})$.

To prove the lower bound also for $p(x) \in o(x^{-2/3})$, consider the graph
built for $p(x) = x^{-2/3}$, and add a chain of $\log_{\alpha}(p(m)/m^{-2/3})$
nodes between the parents of $v_i$ and $v_i$ itself. This damps the contribution
of the parents by a factor $\alpha^{\log_{\alpha}(p(m)/m^{-2/3})} = p(m) /
m^{-2/3}$, guaranteeing $P(v_i) \in \Theta(p(n))$ while preserving the
relative score separation of the target nodes. Clearly, MC solves correctly
this instance if and only if it solves correctly the instance built for $p(x) =
x^{-2/3}$, and the same bound holds.
\end{proof}

\end{document}